\documentclass[letterpaper,11pt]{article}

\usepackage{verbatim, url}
\usepackage{latexsym}
\usepackage{amsmath,amssymb}
\usepackage{epsfig}

\usepackage{fullpage} 

\let\myPushQED=\pushQED
\let\myPopQED=\popQED
\newcommand{\myignore}[1]{}
\newenvironment{proof*}
  {\let\pushQED=\myignore\begin{proof}\let\pushQED=\myPushQED}
  {\def\popQED{}\end{proof}\let\popQED=\myPopQED}

{\makeatletter
 \gdef\xxxmark{%
   \expandafter\ifx\csname @mpargs\endcsname\relax 
     \expandafter\ifx\csname @captype\endcsname\relax 
       \marginpar{xxx}
     \else
       xxx 
     \fi
   \else
     xxx 
   \fi}
 \gdef\xxx{\@ifnextchar[\xxx@lab\xxx@nolab}
 \long\gdef\xxx@lab[#1]#2{{\bf [\xxxmark #2 ---{\sc #1}]}}
 \long\gdef\xxx@nolab#1{{\bf [\xxxmark #1]}}
}

\newtheorem{theorem}{Theorem}               
\newtheorem{lemma}[theorem]{Lemma}          

\newcommand{\eps}{\varepsilon}
\newcommand{\twodots}{\mathinner{\ldotp\ldotp}}

\newcommand{\E}{\mathbf{E}}

\newcommand{\req}[1]{(\ref{#1})}
\newcommand{\ceil}[1]{\lceil{#1}\rceil}

\newcommand{\evt}{\mathcal{E}}

\let\phi=\varphi

\def\tsize{t}

\renewcommand{\th}{\ifmmode{^{\textrm{th}}}\else{\textsuperscript{th}\ }\fi}
\newcommand{\nd}{\ifmmode{^{\textrm{nd}}}\else{\textsuperscript{nd}\ }\fi}
\newcommand{\rd}{\ifmmode{^{\textrm{rd}}}\else{\textsuperscript{rd}\ }\fi}

\newcommand{\drop}[1]{}

\newcommand{\qed}{\hbox{\rule{6pt}{6pt}}}
\newenvironment{proof}[1][]{\paragraph{Proof{#1}}}{\hfill\qed\medskip\\}

\title{On the $k$-Independence Required by Linear Probing and Minwise Independence\footnote{A preliminary version of this paper was presented at
{\em The 37th International Colloquium on Automata, Languages and Programming (ICALP'10)} {\cite{patrascu10kwise-lb}}.}} 
\author{
     Mihai P\v{a}tra\c{s}cu\footnote{Passed away June 5, 2012.}
\and Mikkel Thorup\footnote{University of Copenhagen. Research
partly supported by an Advanced Grant from the Danish Council for Independent Research under the Sapere Aude research carrier programme.}}

\begin{document}

\maketitle

\begin{abstract}
We show that linear probing requires 5-independent hash functions for
expected constant-time performance, matching an upper bound of [Pagh
et al.~STOC'07,SICOMP'09]. More precisely, we construct a random 4-independent hash function
yielding expected logarithmic search time for certain keys.
For $(1+\eps)$-approximate minwise independence, we show that
$\Omega(\lg \frac{1}{\eps})$-independent hash functions are required,
matching an upper bound of [Indyk, SODA'99, JALG'01].
We also show that the very fast 2-independent multiply-shift scheme of
Dietzfelbinger [STACS'96] fails badly in both applications.
\end{abstract}

\section{Introduction}
The concept of $k$-independence was introduced by Wegman and
Carter~\cite{wegman81kwise} in FOCS'79 and has been the cornerstone of
our understanding of hash functions ever since. Formally, we think of a hash function $h: [u] \to [\tsize]$ as a random variable distributed over
$[\tsize]^{[u]}$. Here $[s]=\{0,\ldots,s-1\}$. 
We say that $h$ is $k$-independent
if (1) for any distinct keys $x_1, \dots, x_k \in [u]$, the hash values
$h(x_1), \dots, h(x_k)$ are independent random variables; and (2) for
any fixed $x$, $h(x)$ is uniformly distributed in $[\tsize]$.

As the concept of independence is fundamental to probabilistic
analysis, $k$-independent hash functions are both natural and powerful in
algorithm analysis. They allow us to replace the heuristic assumption
of truly random hash functions that are uniformly distributed in $[\tsize]^{[u]}$, hence
needing $u\lg \tsize$ random bits ($\lg=\log_2$), with real implementable hash
functions that are still ``independent enough'' to yield provable
performance guarantees similar to those proved with true randomness. We are then left with the natural goal of
understanding the independence required by algorithms. 

Once we have proved that $k$-independence suffices for a hashing-based 
randomized algorithm, we are free to use {\em any\/} $k$-independent hash function.
The canonical construction of a $k$-independent hash function is based on
polynomials of degree $k-1$. Let $p \ge u$ be prime. Picking random
$a_0, \dots, a_{k-1} \in \{0, \dots, p-1\}$, the hash function is
defined by:
\[ h(x) = \Big( \big( a_{k-1} x^{k-1} + \cdots + a_1 x + a_0 \big)
        \bmod{p} \Big)\]
If we want to limit the range of hash values to $[t]$, we use
$h(x)\bmod t$. This preserves requirement (1) of independence among
$k$ hash values. Requirement (2) of uniformity is close to 
satisfied if $p\gg t$.

Sometimes 2-independence suffices. For example, 2-independence implies
so-called universality \cite{carter77universal}; namely that the probability
of two keys $x$ and $y$ colliding with $h(x)=h(y)$ is $1/\tsize$; or close
to $1/t$ if the uniformity of (2) is only approximate.
Universality implies expected constant time performance of hash tables
implemented with chaining. Universality also suffices
for the 2-level hashing of Fredman et al.
\cite{fredman84dict}, yielding static hash tables with constant query time.

At the other end of the spectrum, when dealing with problems involving $n$ objects, $O(\lg n)$-independence suffices in
a vast majority of applications. One reason for this is the Chernoff
bounds of~\cite{schmidt95chernoff} for $k$-independent events, whose
probability bounds differ from the full-independence Chernoff bound by
$2^{-\Omega(k)}$. Another reason is that random graphs with $O(\lg
n)$-independent edges~\cite{alon08kwise} share many of the properties
of truly random graphs.

The independence measure has long been central to the study of
randomized algorithms. It applies not only to hash functions, but also 
to pseudo-random number generators viewed as assigning hash values to $0,1,2,..$.  For example, \cite{karloff93prg} considers variants of
QuickSort, \cite{alon99linear} consider the maximal bucket size for
hashing with chaining, and   
\cite{cohen09cuckoo5,dietzfelbinger09cuckoo-bas} consider Cuckoo hashing.
In several cases \cite{alon99linear,dietzfelbinger09cuckoo-bas,karloff93prg},
it is proved that linear transformations $x\mapsto \big( (ax + b) \bmod p \big)$ do not suffice for good performance, hence that 
2-independence is not in itself sufficient.

In this paper, we study the independence for two important
applications in which it is already known that 2-independence does not
suffice: linear probing and minwise-independent hashing.

\subsection{Linear probing} 
Linear probing is a classic implementation of hash tables. It uses a
hash function $h$ to map a set of $n$ keys into an array of size
$\tsize$. When inserting $x$, if the desired location $h(x)\in
[\tsize]$ is already occupied, the algorithm scans $h(x)+1, h(x)+2,
\dots,\tsize-1,0,1,\ldots$ until an empty location is found, and
places $x$ there. The query algorithm starts at $h(x)$ and scans
either until it finds $x$, or runs into an empty position, which
certifies that $x$ is not in the hash table.  When the query search is
unsuccessful, that is, when $x$ is not stored, the query algorithm
scans exactly the same locations as an insert of $x$. A general bound
on the query time is hence also a bound on the insertion time.

We generally assume constant load of the hash table, e.g.~the number
of keys is $n \le \frac{2}{3}\tsize$.

This classic data structure is one of the most popular implementations of hash
tables, due to its unmatched simplicity and efficiency. 
The practical use of linear probing dates back at least to 1954 to an
assembly program by Samuel, Amdahl, Boehme (c.f. \cite{knuth-vol3}).
On modern
architectures, access to memory is done in cache lines (of much more
than a word), so inspecting a few consecutive values typically
translates into a single memory access. Even if the scan straddles a
cache line, the behavior will still be better than a second random
memory access on architectures with prefetching. Empirical
evaluations~\cite{black98linprobe,heileman05linprobe,pagh04cuckoo}
confirm the practical advantage of linear probing over other known
schemes, while cautioning~\cite{heileman05linprobe,thorup12kwise} that
it behaves quite unreliably with weak hash functions (such as
2-independent). Taken together, these findings form a strong
motivation for theoretical analysis.

Linear probing was shown to take expected constant time for
any operation in 1963 by Knuth~\cite{knuth63linprobe}, in a report which is now
regarded as the birth of algorithm analysis. This analysis, however, 
assumed a truly random hash function. 

A central open question of Wegman and Carter~\cite{wegman81kwise} was
how linear probing behaves with $k$-independence. Siegel and
Schmidt~\cite{schmidt90hashing,siegel95hashing} showed that $O(\lg
n)$-independence suffices for any operation to take expected constant time. 
Pagh et al.~\cite{pagh07linprobe} showed that just $5$-independence
suffices for this expected constant operation time. 
They also showed that linear transformations do not suffice, hence that 2-independence is not in itself sufficient.

Here we close this line
of work, showing that $4$-independence is not in itself sufficient for
expected constant operation time. We display
a concrete combination of keys and a 4-independent random hash function
where searching certain keys takes super constant expected time. This
shows that the $5$-independence result of Pagh et al.~\cite{pagh07linprobe} 
is best possible.

\begin{table}[t]
  \centering
  \begin{tabular}{|l|c|c|c|c|}
    \hline
    Independence & 2 & 3 & 4 & $\ge 5$ \\
    \hline
    Query time & $\Theta(\sqrt{n})$
              & $\Theta(\lg n)$ & $\Theta(\lg n)$ & $\Theta(1)$ \\
    \hline
    Construction time & $\Theta(n\lg n)$ & $\Theta(n\lg n)$ 
              & $\Theta(n)$ & $\Theta(n)$ \\
    \hline
  \end{tabular}
\caption{Expected time bounds for linear probing with a poor $k$-independent
  hash function. The bounds are worst-case expected, e.g., a lower bound
for the query means that there is a concrete combination of stored set, query 
key, and $k$-independent hash function with this expected search time while the upper-bound means that this is the worst expected time for any such combination. Construction time refers to the worst-case expected total time for inserting $n$   keys starting from an empty table.}
\label{tab:lin-probe}
\end{table}

We will, in fact, provide a complete 
understanding of linear probing with low independence as summarized in
Table~\ref{tab:lin-probe}. This includes a new upper and lower bound
of $\Theta(\sqrt n)$ for the query time with 2-independence. All the
other upper bounds in the table are contained, at least implicitly, in
\cite{pagh07linprobe}. On the lower bound side, the only
lower bound known from \cite{pagh07linprobe} was the $\Omega(n\log n)$ lower bound on the construction
time with 2-independence, which we show here also holds with 3-independence.

\subsection{Minwise independence}
The concept of minwise independence was introduced by two classic algorithms: detecting
near-duplicate documents~\cite{broder98minwise,broder97minwise} and
approximating the size of the transitive
closure~\cite{cohen97minwise}.  The basic step in these algorithms is
estimating the size of the intersection of pairs of sets, relative to
their union: for $A$ and $B$, we want to estimate $\frac{|A\cap B|}{|A
  \cup B|}$ (the {\em Jaccard similarity coefficient}). To do this
efficiently, one can choose a hash function $h$ and maintain $\min
h(A)$ as the sketch of an entire set $A$. If the hash function is
truly random, we have $\Pr[\min h(A) = \min h(B)] = \frac{|A\cap
  B|}{|A \cup B|}$. Thus, by repeating with several hash functions, 
the Jaccard
coefficient can be estimated up to a small error.

To make this idea work, the property required of the hash
function is \emph{minwise independence}. Formally, a random hash
function $h: [u] \to [u]$ is said to be minwise
independent if, for any set $S \subset [u]$ and any $x \notin S$, we
have $\Pr_{h}[ h(x) < \min h(S) ] = \frac{1}{|S|+1}$.  In
other words, $x$ is the minimum of $S \cup \{x\}$ only with its
``fair'' probability $\frac{1}{|S|+1}$.  

A hash function providing a truly random permutation on $[u]$ is
minwise independent, but representing such a function
requires $\Theta(u)$ bits \cite{broder98minwise}. Therefore the definition is relaxed to
\emph{$\eps$-minwise independent}, requiring that $\Pr_{h}[ h(x) < \min h(S) ]
= \frac{1 \pm \eps}{|S|+1}$.  Using such a function, we will have
$\Pr[\min h(A) = \min h(B)] = (1\pm \eps) \frac{|A\cap B|}{|A \cup
  B|}$. Thus, the $\eps$ parameter of the minwise independent hash function 
dictates the best approximation achievable in the algorithms (which
\emph{cannot} be improved by repetition).

Broder et al. \cite{broder98minwise} proved that 
linear transformations are only $\Omega(\log n)$-minwise independent.
Indyk~\cite{indyk01minwise} provided the only known implementation of
minwise independence with provable guarantees, showing that $O(\lg
\frac{1}{\eps})$-independent hash functions are $\eps$-minwise
independent.  

In this paper, we show for any $\eps>0$, that there exist $\Omega(\lg
\frac{1}{\eps})$-independent hash functions which are no better than
$\eps$-minwise independent, hence that Indyk's result is best possible.

\subsection{Concrete Schemes.}
Our results provide a powerful understanding of a natural combinatorial
resource (independence) for two important algorithmic questions. In
other words, they provide limits on how far the \emph{paradigm} of
independence can take us. Note, however, that independence is only
one of many properties that concrete hash schemes can possess. In a particular
application, a hash scheme can behave much better than its
independence guarantees, if it has some other probabilistic property
unrelated to independence. Obviously, proving that a concrete hashing
scheme works is not as attractive as proving that every
$k$-independent scheme works, including more efficient
$k$-independent schemes found in the future. However, if low independence
does not work, then a concrete scheme may be the best we can hope for.

For both of our applications, we know that the classic 
linear transformation $x\mapsto \big( (ax + b) \bmod p \big)$ 
does not give good bounds~\cite{broder98minwise,pagh07linprobe}.
However, there is a much more practical 2-independent hash function; namely Dietzfelbinger's
multiply-shift scheme~\cite{dietzfel96universal}, which on some
computers is 10 times as fast~\cite{thorup00universal}. To hash
$w$-bit integers to $\ell$-bit integers, $\ell\leq w$, the scheme
picks two random $2w$-bit integers $a$ and $b$, and maps
$x\mapsto$\texttt{(}$a$\texttt{*}$x$\texttt{+}$b$\texttt{)>>}$(2w-\ell)$.
The operators are those from the programming language C \cite{KR88},
where \texttt{*} and \texttt{+} are $2w$-bit multiplication and addition,
and \texttt{>>} is an unsigned right shift. 

We are not aware of any previous papers considering the concrete limits
of multiply-shift in concrete applications, but in this paper, we prove
that linear probing with multiply-shift hashing suffers from
$\Omega(\lg n)$ expected operation times on some inputs. Similarly, we
show that minwise independent hashing may have a very large
approximation error of $\eps = \Omega(\lg n)$. These lower bounds
match those from~\cite{broder98minwise,pagh07linprobe} for the classic
linear transformations, and may not be surprising given the ``moral
similarity'' of the schemes, but they do require different rather
involved arguments. We feel that this effort to understand the limits
of multiply-shift is justified, as it brings the theoretical lower
bounds more in line with programming reality.

\paragraph{Later work.} 
After the negative findings of the current paper, we continued our
quest for concrete hashing schemes that were both efficient and possessed
good probabilistic properties for our target applications. We
considered simple tabulation hashing~\cite{patrascu12charhash}, which
breaks fundamentally from polynomial hashing schemes. Tabulation based
hashing is comparable in speed to multiply-shift
hashing~\cite{dietzfel96universal}, but it uses much more space
($u^{\Omega(1)}$ where $u$ is the size of the universe instead of
constant). Simple tabulation is only 3-independent, yet it does give
constant expected operation time for linear probing and $o(1)$-minwise
hashing. We also proposed a variant, twisted tabulation, with even
stronger probabilistic guarantees for both linear probing and minwise
hashing \cite{DT14:twist-min,PT13:twist}. Both of these tabulation schemes
are of a general nature with many applications even though they are 
only 3-independent.

We note that there has been several other studies of hashing schemes that for
other concrete applications have greater power than their independence
suggests, e.g., \cite{aumullerDW14:cuckoo-hash,dietzfel03tabhash}. 
The focus in this paper, however, is hashing for linear probing and minwise hashing. 

The problems discovered here for minwise hashing, also made the last
author consider an alternative to repeating minwise hashing $d$ times
independently; namely to store the $d$ smallest hash value with a
single hash function \cite{thorup13bottomk}. It turns out that as $d$
increases, this scheme performs well even with 2-independence.

\section{Linear probing with $k$-independence}

To better situate our lower bounds, we will first present some simple
proofs of the known upper bounds for linear probing from 
Table~\ref{tab:lin-probe}. This is the
$O(n\lg n)$ expected construction time with 2-independence, 
the $O(n)$ expected construction time with 4-independence,
the $O(\lg n)$ expected query time with 3-independence, and
the $O(1)$ expected query time with 5-independence. The last
bound is the main result from~\cite{pagh07linprobe}, and all the
other bounds are at least implicit in~\cite{pagh07linprobe}.
Our proof here is quite different from that in~\cite{pagh07linprobe}: 
simpler and more close in line
with our later lower bound constructions. Our proof is
also simplified in that we only consider load factors below $2/3$. A more
elaborate proof obtaining tight bounds for 5-independence for all load factors $1-\eps$ is presented in~\cite{patrascu12charhash}.

The main probabilistic tool featuring
in the upper bound analysis is standard moment bounds: consider throwing $n$ balls
into $b$ bins uniformly. Let $X_i$ be the indicator variable for the
event that ball $i$
lands in some fixed bin, and $X = \sum_{i=1}^n X_i$ the number of balls
in the this bin. We have $\mu = \E[X] = \frac{n}{b}$. As usual, the
$k$\th central moment of $X$ is defined as $\E[(X-\mu)^k]$. If $k=O(1)$ and $\mu=\Omega(1)$, then 
$\E[(X-\mu)^k]=O(\mu^{k/2})$. Therefore, by Markov's inequality, if $k$ is
further even, 
\begin{equation}\label{eq:moment}
\Pr[|X-\mu|\geq \alpha \mu]=\Pr[(X-\mu)^{k}\geq
  \alpha^k\mu^k]=O(1/(\alpha^k\mu^{k/2})).
\end{equation}
These $k$\th moment bounds were also used in~\cite{pagh07linprobe}, but
the way we apply them here is quite different. We consider a perfect
binary tree spanning the array $[t]$ where $t$ is a power of two. 
A node at 
height $h\leq \lg_2 t$ has an interval of $2^h$ array positions below it, and
is identified with this interval. 

We assume that the load factor is at most $2/3$, that is $n\leq
\frac23\,\tsize$, so we expect at most $\frac23\, 2^h$ keys to
hash to the interval of a height $h$ node (recall that
with linear probing, keys may end up in positions later than
the ones they hash to). Call the node ``near-full'' if at least $\frac34\,
2^h$ keys hash to its interval.

\paragraph{Construction time for $k=2,4$}
We will now bound the total expected time it takes to construct
the hash table (the cost of inserting $n$ distinct keys). 
A run is a maximal interval of filled locations. If the
table consists of runs of $\ell_1, \ell_2, \dots$ keys ($\sum \ell_i = n$),
the cost of constructing it is bounded from above by $O(\ell_1^2 + \ell_2^2
+ \dots)$. We note that runs of length $\ell_i<4$ contribute $O(n)$ to
this sum of squares. To bound the longer runs, we make the following crucial
observation: if a run contains between $2^{h+2}$ and $2^{h+3}$
keys for $h\geq 0$, then some node at height $h$ above it is near-full.  In
fact, there will be such a near-full height $h$ node whose last
position is in the run.

For a proof, we study a run of length at least $2^{h+2}$.  The run is
preceded by an empty position, so all keys in the run are hashed to
the run (but may appear later in the run than the position they hashed
to).  We now consider the first 4 height $h$ nodes with their
last position in the interval. The last 3 of these have all their positions
in the run.
Assume for a contradiction that none of
these are near-full. The first node (whose first positions may not be
in the run) contributes less than $\frac 34\, 2^h$ keys to the run
(in the most extreme case, this many keys hash to the last position of
that node).  The subsequent nodes have all $2^h$ positions in the
run, but with less than $\frac 34\,2^h$ keys hashing to these
positions. Even with the maximal excess from the first node, we cannot
fill the intervals of three subsequent nodes, so the run must stop
before the end of the third node, contradicting that its last position
was in the run.

Each node has its last position in at most one run, so the observation
gives an upper bound on the cost: for each height $h\geq 0$, add $O(2^{2(h+2)})=O(2^{2h})$ for each
near-full node at height $h$.  Denoting by $p(h)$ the
probability that a node on height $h$ is near-full, the expected
total cost over all heights is thus bounded by 
\[O\left(\sum_{h=0}^{\lg_2 t} (t/2^h)\cdot p(h) \cdot 2^{2h} \right)=
O\left(n\cdot \sum_{h=0}^{\lg_2 t} 2^h\cdot p(h)\right).\] 
Applying  \req{eq:moment} with $\mu= \frac 23\,2^h$
and $\alpha=\frac 34/\frac 23=\frac 98$, we get $p(h)=O(2^{-kh/2})$. With $k=2$,
we obtain $p(h) = O(2^{-h})$, so the total
expected construction cost with 2-independence is $O(n\lg n)$. However, the 4\th moment gives $p(h) =
O(2^{-2h})$, so the total expected construction cost with 4-independence is $O(n)$. These are the upper bounds on the expected construction time for Table~\ref{tab:lin-probe}.

\paragraph{Query time for $k=3,5$}
To bound the running time of one particular operation (query or insert
$q$), we first pick that hash value of $q$. Conditioned
on this choice, the hashing of the stored keys is $(k-1)$-dependent.
The analysis is now very similar to the one for the construction time referring
to the same binary tree.

Suppose the hash of $q$ is contained in a run of length $\ell$. Then
$O(\ell)$ bounds the query time. Assume $\ell\in[2^{h+2}, 2^{h+3})$ for $h\geq 0$.
Then as we argued above, one of the first 4 nodes of height $h$
whose last position is in the run is near-full. Since the run contains
the fixed hash of $q$ and is of length at most $2^{h+3}$, there are at
most $12$ relevant height $h$ nodes; namely the ancestor of the hash of $q$,
the $8$ nodes to its left and the $3$ nodes to its right. Each has probability
$p(h)$ of being near-full, so the expected run length is
\[\E[\ell]\leq 3+\sum_{h=0}^{\lg t}12\cdot p(h)\cdot 2^{h+3}=
O\left(\sum_{h=0}^{\lg t}p(h)\cdot 2^{h}\right).\]
This time, we use $k'=k-1$ in \req{eq:moment}, so with $3$-independence
we obtain $p(h) = O(2^{-h})$, and an expected query time
of $O(\lg n)$. With 5-independence, we get $p(h) =
O(2^{-2h})$, so the expected query time is $O(1)$.

\paragraph{Our results.}
Two intriguing questions pop out of the above analysis. First, is the
independence of the query really crucial? Perhaps one could argue that
the query behaves like an average operation, even if it is not
completely independent of everything else. Secondly, one has to wonder
whether 3-independence suffices (by using something other than 3\rd
moment): all that is needed is a bound slightly stronger than 2\nd
moment in order to make the costs with increasing heights decay geometrically!

We answer both questions in strong negative terms.  The complete
understanding we provide of linear probing with low independence is
summarized in Table~\ref{tab:lin-probe}.  Addressing the first
question, we show that there are 4-independent hash functions that for
certain combinations of query and stored keys lead to an expected
search time of $\Omega(\lg n)$ time. Our proof demonstrates an
important phenomenon: even though most bins have low load, a
particular query key's hash value could be correlated with the
(uniformly random) choice of \emph{which} bins have high load.

An even more striking illustration of this fact happens for
2-independence: the query time blows up to $\Omega(\sqrt{n})$ in
expectation, since we are left with no independence at all after
conditioning on the query's hash. A matching upper bound will also
be presented. This demonstrates a very large
separation between linear probing and collision chaining, which enjoys
$O(1)$ query times even for 2-independent hash functions.

Addressing the second question, we show that 3-independence is not
enough to guarantee even a construction time of $O(n)$. Thus, in some
sense, the 4\th moment analysis is the best one can hope for.

The constructions will be progressively more complicated as the
independence $k$ grows, and the constructions for higher $k$ will
assume a full understanding of the constructions for lower $k$.

\subsection{Expected Query Time $\Theta(\sqrt{n})$ with 2-Independence}\label{sec:2-ind}

Above we saw that the expected construction time with $2$-independence
is $O(n\lg n)$, so the average cost per key is $O(\lg n)$. We will
  now define a 2-independent hash function such that the expected query
  time for some concrete key is $\Omega(\sqrt{n})$. Afterwards,
we will show a matching upper bound of $O(\sqrt{n})$
that holds with any 2-independent
hash function. 

The main idea of the lower bound proof is that a designated query $q$
can play a special role: even if most portions of the hash table are
lightly loaded, the query can be correlated with the portions that
\emph{are} loaded. We assume that the number $n$ of stored keys is a
square and that the table size is $\tsize=2n$.

We think of the stored keys and the query key as given, and we
want to find bad ways of distributing them $2$-independently into the
range $[\tsize]$. To extend the hash function to the entire universe, all
other keys are hashed totally randomly. We consider unsuccessful
searches, i.e.~the search key $q$ is not stored in the hash table. The
query time for $q$ is the number of cells considered from $h(q)$ up to
the first empty cell. If, for some $d$, the interval $Q=(h(q)-d,h(q)]$
has $2d$ or more keys hashing into it, then the search time is $\Omega(d)$.

Let $d=2\sqrt n$, noting that $d$ divides $\tsize$.  In our
construction, we first pick the hash value $h(q)$ uniformly. We then divide
the range into $\sqrt{n}$ intervals of length $d$, of the form $(h(q)
+ i\cdot d,\, h(q) + (i+1)d]$, wrapping around modulo $\tsize$. One of these
intervals is exactly $Q$.

Below we prescribe the distribution of stored keys among the intervals. We will
only specify how many keys go in each interval. Otherwise,
the distribution is assumed to be fully random. Thus it is
understood that the keys are randomly permuted between the intervals and
that the keys in an interval are placed fully randomly within 
that interval.

To place $2d =4\sqrt{n}$ keys in the query interval with constant probability, we
mix two strategies, each followed with a constant probabilities to be
determined:
\begin{description}
\item[$S_1$:] Spread keys evenly, with $\sqrt{n}$ keys in each interval. 
\item[$S_2$:] Consider the query interval $Q$ and pick three random intervals, distinct from $Q$ and each other.
  Place $4\sqrt{n}$ keys in a random one of these 4 intervals, and none in the others.  
  All other intervals than these 4 get $\sqrt{n}$ keys.  

  With probability $1/4$, it is $Q$ that gets
  $4\sqrt{n}=2d$ keys, overloading it by a factor 2. Then, as described above,
  the search time is $\Omega(\sqrt{n})$.
\end{description}
To argue that the distribution is 2-independent with appropriate balancing
between $S_1$ and $S_2$, we need to consider
pairs of two stored keys, and pairs involving the query and one stored
key. 

Consider first the query key $q$ versus a stored key $x$. Given $h(q)$, we want to argue that $x$ is placed uniformly at random in $[t]$.
The key $x$ is
placed uniformly in whatever interval it lands in. With
$S_1$, the distribution among intervals is symmetric, so $x$
is indeed placed uniformly in $[t]$ with $S_1$.
Now consider $S_2$.  Since the three special non-query intervals with
$S_2$ are random, $x$ has the same chance of landing in any non-query
interval. All that remains is to argue that the probability that $x$
lands in the query interval $Q$ is $1/\sqrt n$. This follows because
the expected number of keys in $Q$
is $\sqrt n$ and the $n$ stored keys are treated symmetrically.
The hashing of the query and a stored key is thus independent both
with $S_1$ and $S_2$.

We now consider two stored keys. 
We will think of the
hash value $h(q)$ as being picked in two steps. First we pick
the offset $r(q)=h(q)\bmod d$ uniformly at random. This offset decides the
locations of our intervals as $(r(q)+(j-1)\cdot d,\,r(q)+j\cdot d]$,
for $j=0,\ldots,\sqrt n-1$, with wrap-around modulo $t$.  
Second with pick the uniformly random index $i(q)=\lfloor h(q)/d\rfloor$
of the query interval $Q=(r(q)+(i(q)-1)\cdot d,\,r(q)+i(q)\cdot d]$.

Now consider the strategy $S_2$ after the offset has been fixed. The
query interval is chosen uniformly at random, so from the perspective
of stored keys, the four special intervals with $S_2$ are
completely random. This means that from the perspective
of the stored keys, all intervals are symmetric both with $S_1$ and $S_2$.

All that remains is to understand the probability of the two keys
landing in the same interval. We call this a ``collision''. We need
to balance the strategies so that the collision probability is exactly
$1/\sqrt n$. Since all stored keys are treated symmetrically, this
is equivalent to saying that the expected number of
collisions among stored keys is $\binom{n}{2}/
\sqrt{n}= \frac{1}{2} n^{1.5} -
\frac{1}{2} \sqrt{n}$.

In strategy $S_1$, we get the smallest possible number of collisions:
$\sqrt{n} \binom{\sqrt{n}}{2} = \frac{1}{2} n^{1.5} - \frac{1}{2}
n$. This is too few by almost $n/2$. In strategy
$S_2$, we get $(\sqrt{n} - 4) \binom{\sqrt{n}}{2} +
\binom{4\sqrt{n}}{2} = \frac{1}{2} n^{1.5} + \frac{11}{2} n$
collisions, which is too much by a bit more than $5.5n$. To get the
right expected number of collisions, we use $S_1$ with probability
$P_{S_1}=\frac{5.5n+0.5\sqrt n}{0.5 n + 5.5n} = \frac{11}{12} +\frac1{12\sqrt n}$. With this mix of strategies, our hashing of keys is 2-independent,
and since $P_{S_2}=\Omega(1)$, the expected search cost is $\Omega(\sqrt n)$.

\paragraph{Upper bound}
We will now prove a matching upper bound of $O(\sqrt n)$ on the expected
query time $T$ with any 2-independent scheme. As an upper bound on the query
time, we consider the longest run length $L$ in the whole linear probing table.
Then $T=O(L)$ no matter which location the query key hash to. Therefore it
does not matter if the hash value of the query key depends on the
hashing of the stored keys.

The table size is $t=(1+\eps)n$, for some $\eps\in(0,1]$, and we assume for simplicity
that $n$ is a square and $\sqrt n$ divides $t$. We will prove
that $\E[L]=O(\sqrt n/\eps)$. As in the lower bound, we divide $[t]$ into 
$\sqrt n$ equal sized intervals. We view
keys as colliding if they hash to the same interval.  We want to argue
that a large run imply too many collisions for $2$-independence, but the argument
is not based on a standard 2\nd moment bound.

Let $C$ be the number of collisions. 
The expected number of collisions is $\E[C]={n \choose 2}/\sqrt n=n^{3/2}/2-
n^{1/2}/2$. The minimum number of collisions is with the distribution
$S_1$ from the lower bound: a perfectly regular distribution with
$n/\sqrt n=\sqrt n$ keys in each interval, hence $C\geq 
\sqrt n\cdot {\sqrt
  n \choose 2}=n^{3/2}/2-n/2$ collisions.

An interval with $m$ keys has $m^2/2-m/2$ collisions and the
derivative is $m-1/2$. It follows that if we move a key from an
interval with $m_1$ keys to one with $m_2\geq m_1$ keys, the number of
collisions increases by more than $m_2-m_1$. Any distribution can be
obtained from the above minimal distribution by moving keys from
intervals with at most $\sqrt n$ keys to intervals with at least
$\sqrt n$ keys, and each such move increases the number of collisions.

A run of length $L$ implies that this many keys hash to an interval of
this length. The run is contained in less than $L/((1+\eps)\sqrt n)+2$ of our
length $t/\sqrt n=(1+\eps)\sqrt n$ intervals.  In the process of creating a
distribution with this run from the minimum distribution, we have to
move at least $L-(\eps\sqrt n/2)(L/((1+\eps)\sqrt n)+2)$ keys to intervals that
already have $\eps\sqrt n/2$ keys added, and each such move
gains at least
$\eps\sqrt n/2$ collisions. Thus our total gain in collisions is at least
\begin{align*}
(\eps\sqrt n/2)(L-(\eps\sqrt n/2)(L/((1+\eps)\sqrt n)+2))&= 
(1-\eps/(2(1+\eps))) \eps L\sqrt n/2-\eps^2 n/2\\
&\geq 3\eps L\sqrt n/8-\eps^2 n/2.
\end{align*}
The total number of collisions $C$ with a run of length $L$ is
therefore at least
\[C_L=n^{3/2}/2-n/2+3\eps L\sqrt n/8-\eps^2n/2\geq n^{3/2}/2-n+3\eps^2L\sqrt n/8.\]
Since $C_L$ is linear in $L$, the expected number of collisions is thus lower bounded by
\[\E[C]\geq\E[C_L]=n^{3/2}/2-n+3\eps\E[L]\sqrt n/8.\]
But $\E[C]=n^{3/2}/2- n^{1/2}/2$, so we conclude that
\[n^{3/2}/2- n^{1/2}/2\geq n^{3/2}/2-n+3\eps\E[L]\sqrt n/8
\implies \E[L]\leq 8(n- n^{1/2})/(3\eps \sqrt n)<3\sqrt n/\eps.\]
The expected maximal run length is thus less than $3\sqrt n/\eps$, so
the expected query time is $O(\sqrt n/\eps)$. Summing up, we have proved
\begin{theorem}\label{thm:2-ind-query} If $n$ keys are stored in 
a linear probing table of size $t=(1+\eps)n$ using a 2-independent scheme, 
then the expected query time for any key is $O(\sqrt n/\eps)$. Moreover,
for any set of $n$ given keys plus a distinct query key,
there exists a 2-independent scheme such that if it
is used to insert the  $n$ keys in a linear probing table of size 
$t=2n$, then the query takes $\Omega(\sqrt n)$ time.
\end{theorem}

\subsection{Construction Time $\Omega(n\lg n)$ with 3-Independence} 
   \label{sec:3wise}

We will now construct a 3-independent hash function, such
that the time to insert $n$ keys into a hash table is $\Omega(n \lg
n)$. The lower bound is based on overflowing
intervals.
\begin{lemma}\label{lem:overflow-constr}
Suppose an interval $[a,b]$ of length $d$ has $d+\Delta$ stored keys hashing
to it. Then the insertion cost of these keys is
$\Omega(\Delta^2)$. 
\end{lemma}
\begin{proof}
The overflowing $\Delta$ keys will be part of a run containing
$(b,b+\Delta]$. At least $\lceil
\Delta/2\rceil$ of them must end at position $b+\lceil
\Delta/2\rceil$ or later, i.e., a displacement of at least $\lceil
\Delta/2\rceil$. Interference from stored keys hashing outside $[a,b]$ can
only increase the displacement, so the insertion cost is $\Omega(\Delta^2)$.
\end{proof}
We will add up such squared
overflow costs over disjoint intervals, demonstrating an expected total cost
of $\Omega(n\lg n)$.

As before, we assume the array size $\tsize=2^p$ is a power of two,
and we set $n = \lceil \frac{2}{3} t \rceil$. We imagine a perfect binary 
tree of height $p$ spanning $[t]$: The root is level $0$
and level $\ell$ is the nodes at depth $\ell$. The $2^p$ leaves on 
level $p$ are identified with $[t]$.

Our hash function
will recursively distribute keys from a node to its two children,
starting at the root. Nodes run independent random distribution
processes. Then, if each node makes a $k$-independent distribution,
overall the function is $k$-independent.

For a node, we mix between two strategies for distributing $2m$ keys between
the two children (here $m$ may only be half-integral):
\begin{description}
\item[$S_1$:] Distribute the keys evenly between the children. If $2m$ is
odd, a random child gets $\lceil m\rceil$ keys. The keys
are randomly permuted, so it is random which keys ends in which interval.
\item[$S_2$:] Give all the keys to a random child.
\end{description}
Our goal is to prove that there is a probability for the second
strategy, $P_{S_2}$, such that the distribution process is
3-independent. Then we will calculate the cost it induces on linear
probing. First, however, we need some basic facts about $k$-independence.

\subsubsection{Characterizing $k$-Independence}

Our distribution procedure treats keys symmetrically, and ignores the
distinction between left/right children. We call such distributions
\emph{fully symmetric}. As above, we consider a node that
has to distribute $2m$
keys to its two children. The key set is identified with $[2m]$.
Let $X_a$ be the
indicator random variable for key $a$ ending in the left child, and
$X=\sum_{a\in[2m]} X_a$. By symmetry of the children, $\E[X_a] = \frac{1}{2}$,
so $\E[X] = m$.  The $k$\th moment is $F_k = \E[(X-m)^k]$. Also define
$p_k = \Pr[X_1 = \dots = X_k = 1]$. Note here by symmetry that 
any $k$ distinct keys yield the same value. Also, by symmetry, $p_1=1/2$.

\begin{lemma} \label{lem:pk}
A fully symmetric distribution is $k$-independent iff $p_i = 2^{-i}$
for all $i=2, \dots, k$.
\end{lemma}

\begin{proof}
For the non-trivial direction, assume $p_i = 2^{-i}$ for all $i=2, \dots, k$.
We need to show that, for any $(x_1, \dots, x_k) \in \{0,1\}^k$,
$\Pr[(X_1=x_1) \land \dots \land (X_k=x_k)] = 2^{-k}$. By symmetry of
the keys, we can sort the vector to $x_1 = \dots = x_t = 1$ and
$x_{t+1} = \dots = x_k = 0$. Let $p_{k,t}$ be the probability that
such a vector is seen.

We use induction on $k$. In the base case, $p_{1,0} = p_{1,1} =
\frac{1}{2}$ by symmetry. For $k\ge 2$, we start with $p_{k,k} = p_k =
2^{-k}$. We then use induction for $t=k-1$ down to $t=0$.  The
induction step is simple: $p_{k,t} = p_{k-1,t} - p_{k,t+1} =
2^{-(k-1)} - 2^{-k} = 2^{-k}$. Indeed, $\Pr[X_1=\cdots=X_t = 1 \land
  X_{t+1}=\cdots= X_k = 0]$ can be computed as the difference between
$\Pr[X_1=\cdots= X_t = 1 \land X_{t+1} =\cdots= X_{k-1} = 0]$ (measured by
$p_{k-1, t}$) and $\Pr[X_1=\cdots= X_t = 1 \land X_{t+1} =\cdots= X_{k-1}
  = 0 \land X_k=1]$ (measured by $p_{k,t+1}$).
\end{proof}

Based on this lemma, we can also give a characterization based on
moments. First observe that any odd moment is necessarily zero, as
$\Pr[X=m+\delta] = \Pr[X=m-\delta]$ by symmetry of the children.

\begin{lemma}\label{lem:Fk}
A fully symmetric distribution is $k$-independent iff its even moments
up to $F_k$ coincide with the moments of the truly random
distribution.
\end{lemma}

\begin{proof}
We will show that $p_2, \dots, p_k$ are determined by $F_2, \dots,
F_k$, and vice versa. Thus, any distribution that has the same moments
as a truly random distribution, will have the same values $p_2, \dots,
p_k$ as the truly random distribution ($p_i = 2^{-i}$ as in
Lemma~\ref{lem:pk}).

Let $n^{\underline{k}} = n (n-1) \dots (n-k+1)$ be the falling
factorial. The complete dependence between $p_2, \dots, p_k$ and $F_2, \dots,
F_k$ follows inductively from the following statement:
\begin{equation}\label{eq:Fk}
F_k = (2m)^{\underline k}p_k  ~+~ f_k(m, p_2,..,p_{k-1}),\qquad
    \textrm{for some function $f_k$.}
\end{equation}
To see this, first note that 
\begin{equation}\label{eq:Fk1}
F_k = \E[(X-m)^k] = \sum_{j=0}^k {k\choose j} \E[X^j](-m)^{k-j}=
\E[X^k] + g_k\big( m,
\E[X^2], \dots, \E[X^{k-1}] \big)
\end{equation}
for some function $g_k$. Moreover,  
\begin{equation}\label{eq:Fk2}
\E[X^k] = \sum_{(a_1,\ldots,a_k)\in [2m]^k}\E[X_{a_1}\cdots X_{a_k}]
= d_0(m,k) p_k+d_1(m,k) p_{k-1}+ \cdots +d_{k-1}(m,k)p_1\textnormal, 
\end{equation}
where $d_i(m,k)$ is the number of tuples $(a_1,\ldots,a_k)\in [2m]^k$
with $i$ duplicates, hence $k-i$ distinct keys. 
In particular, $d_0(m,k)=(2m)^{\underline k}$, and by symmetry, we always
have $p_1=1/2$. Combining this with \req{eq:Fk1} and \req{eq:Fk2}, we get that 
\[F_k = (2m)^{\underline k} p_k+g_k^*\big( m,p_2,..,p_{k-1},
\E[X^2], \dots, \E[X^{k-1}] \big)\]
for some function $g^*_k$, and then \req{eq:Fk} follows by
induction.
\end{proof}

\subsubsection{Mixing the Strategies} 
As a general convention, when we are mixing strategies $S_i$, we
use $P_{S_i}$ to denote the probability of picking strategy $S_i$ while we
use a superscript ${S_i}$ to denote measures within strategy $S_i$, e.g.,
$F_2^{S_i}$ is the second moment when strategy $S_i$ is applied.

Our strategies $S_1$ and $S_2$ are both fully symmetric, so by 
Lemma \ref{lem:Fk}, a mix of $S_1$ and $S_2$ is
3-independent iff it has the correct 2\nd moment $F_2 =
\frac{m}{2}$. In strategy $S_1$, $X = m \pm 1$ (due to rounding errors
if $2m$ is odd), so $F_2^{S_1} \le 1$. In $S_2$ (all to one child),
$|X-m|=m$ so $F_2^{S_2}=m^2$. For a correct 2\nd moment of $m/2$, we balance
with $P_{S_2} =\frac{1}{2m} \pm O(\frac{1}{m^2})$.

\subsubsection{The Construction Cost of Linear Probing}
We now calculate the cost in terms of squared overflows. As long as
the recursive steps spread the keys evenly with $S_1$, the load factor
stays around $2/3$: at level $\ell$, the intervals have length
$t/2^\ell$ and $2m=2/3\cdot t/2^\ell\pm 1$ keys to be split between
child intervals of length $t/2^{\ell+1}$.  If now, for a node
$v$ on level $\ell$, we apply strategy $S_2$ collecting all keys into
one child, that child interval gets an overflow of $1/3\cdot
n/2^\ell\pm 1=\Omega(m)$ keys.  By Lemma \ref{lem:overflow-constr},
the keys at the child will have a total insertion cost of
$\Omega(m^2)$. Since $P_{S_2} = \Theta(1/m)$, the expected cost
induced by $v$ is $\Omega(m)=\Omega(n/2^i)$.

The above situation is the only one in which we will charge keys at
a node $v$, that is, the keys at $v$ are only charged if the $S_2$
collection is applied to $v$ but to no ancestors of $v$. This
implies that the same key cannot be charged at different nodes.
In fact, we will only charges nodes $v$ at the top $(\lg n)/2$ levels
where the chance that the $S_2$ collection has been done higher
up is small.

It remains to bound the probability that the $S_2$ collection has been
applied to an ancestor of a node $v$ on a given level $\ell\leq (\lg
n)/2$. The collection probability for a node $u$ on level $i\leq\ell$
is $P_{S_2}=\Theta(1/m)=\Theta(2^i/n)$ assuming no collection among
the ancestors of $u$. By the union bound, the probability that any
ancestor $u$ of $v$ is first to be collected is $\sum_{i=0}^{\ell-1}
\Theta(2^i/n)=\Theta(2^\ell/n)=\Theta(1/\sqrt n)=o(1)$.  We conclude
that $v$ has no collected ancestors with probability $1-o(1)$, hence
that the expected cost of $v$ is $\Omega(n/2^\ell)$ as above. The
total expected cost over all $2^\ell$ level $\ell$ nodes is thus
$\Omega(n)$. Summing over all levels $\ell\leq (\lg n)/2$, we get an
expected total insertion cost of $\Omega(n\lg n)$ for our
3-independent scheme. Thus we have proved
\begin{theorem}\label{thm:2-ind-insert} 
For any set of $n$ given keys,
there exists a 3-independent hashing scheme such that if it
is used to insert the  $n$ keys in a linear probing table of size 
$t=2n$, then the expected total insertion time is $\Omega(n\log n)$.
\end{theorem}

\subsection{Expected Query Time $\Omega(\lg n)$ with 4-Independence}
   \label{sec:4wise}

Proving high expected search cost with
4-independence combines the ideas for 2-independence and
3-independence. However, some quite severe complications will arise. 
The lower bound is based on overflowing
intervals.
\begin{lemma}\label{lem:overflow-query}
Suppose an interval $[a,b]$ of length $d$ has $d+\Delta$, 
$\Delta=\Omega(d)$, stored keys hashing
to it. Assuming that the interval has even length and that
the stored keys hash symmetrically to the first
and second half of $[a,b]$. Moreover, assume that the query key hashes uniformly
in $[a,b]$. Then the expected query time is $\Omega(\Delta)$.
\end{lemma}
\begin{proof}
By symmetry between the first and the second half, with probability $1/2$, the first half gets half the keys, hence an overflow of
$\Delta/2$ keys, and a run containing $[a+d/2,a+d/2+\Delta/2)$. Since
$\Delta=\Omega(d)$, the probability that the query key hits the first
half of this run is $\Omega(1)$, and then the expected query cost is $\Omega(\Delta)$.
\end{proof}
As for
2-independence, we will first choose $h(q)$ and then make the stored
keys cluster preferentially around $h(q)$. As for 3-independence, the
distribution will be described using a perfectly balanced binary tree
over $[t]$. The basic idea is to use the 3-independent
distribution from Section~\ref{sec:3wise} along the query path. For brevity, we call nodes on the
query path query nodes. The overflows that lead to an
$\Omega(n\lg n)$ construction cost, will yield an $\Omega(\lg n)$
expected query time. However, the clustering of our 3-independent distribution is far too strong for
4-independence, and therefore we cannot apply it in the
top of the tree. However, further down, we can balance clustering
on the query path by some anti-clustering 
distributions outside the query path.

\subsubsection{3-independent Building Blocks}
For a node that has $2m$
keys to distribute, we consider three basic strategies:
\begin{description}
\item[$S_1$:] Distribute the keys evenly between the two children. If $2m$ is
odd, a random child gets $\lceil m\rceil$ keys.

\item[$S_2$:] Give all the keys to a random child.

\item[$S_3$:] Pick a child randomly, and give it $m + \delta = \lceil m +
  \sqrt{m/2} \rceil$ keys.
\end{description}

By mixing among these, we define two super-strategies:
\begin{description}
\item[$T_1=$] $P_{S_2} \times S_2 ~+~ (1-P_{S_2}) \times S_1$;
\item[$T_2=$] $P_{S_3} \times S_3 ~+~ (1 - P_{S_3}) \times S_1$.
\end{description}
The above notation means that strategy $T_1$ picks strategy $S_2$ with
probability $P_{S_2}$; $S_1$ otherwise. Likewise $T_2$ picks 
$S_3$ with probability $P_{S_3}$; $S_1$ otherwise.
The probabilities $P_{S_2}$ and $P_{S_3}$ are chosen such that $T_1$
and $T_2$ are 3-independent. The strategy $T_1$ is the
3-independent strategy from Section~\ref{sec:3wise} where we
determined $P_{S_2} = \frac{2}{m} \pm O(\frac{1}{m^2})$. This
will be our preferred strategy on the query path.

To compute $P_{S_3}$, we employ the 2\nd moments: $F_2^{S_1} \le 1$ and
$F_2^{S_3}= \frac{m}{2} + O(\sqrt{m})$. (If one
ignored rounding, we would have the precise bounds $F_2^{S_1} = 0$ and
$F_2^{S_3} = \frac{m}{2}$.) By Lemma~\ref{lem:Fk}, we need a 2\nd
moment of $m/2$. Thus, we have $P_{S_3} = 1 - O(\frac{1}{\sqrt{m}})$.

\subsubsection{4-Independence on the Average, One Level At The Time}
We are going to get 4-independence by an appropriate mix of our
3-independent strategies $T_1$ and $T_2$. Our first step is to hash the
query uniformly into $[t]$. This defines the query path. We will do the mixing
top-down, one level $\ell$ at the time. The individual node will not
distribute its keys 4-independently.  Nodes on the
query path will prefer $T_1$ while keys outside the query path will
prefer $T_2$, all in a mix that leads to global 4-independence. 
There will also be neutral nodes for which we use a truly random distribution.
Since all distributions are 3-independent regardless of the query path,
the query hashes independently of any 3 stored keys. We are therefore only concerned about the 4-independence among stored keys.

It is tempting to try balancing of $T_1$ and $T_2$ via 4\th
moments using Lemma~\ref{lem:Fk}. However, even on the same level
$\ell$, the distribution of the number of keys at the node on the
query path will be different from the distributions outside the query
path, and this makes balancing via 4\th moments non-obvious. Instead, we will argue independence via
Lemma~\ref{lem:pk}: since we already have 3-independence and all
distributions are symmetric, we only need to show $p_4 = 2^{-4}$.
Thus, conditioned on 4 given keys $a,b,c,d$ being together on level
$\ell$, we want them all to go to the left child with probability
$2^{-4}$. By symmetry, our 4-tuple $(a,b,c,d)$ is uniformly random
among all 4-tuples surviving together on level $\ell$. On the average
we thus want such 4-tuples to go left together with probability
$2^{-4}$.

\subsubsection{Analyzing $T_1$ and $T_2$}
Our aim now is to compute $p_4^{T_1}$ and
$p_4^{T_2}$ for a node with $2m$ keys to be split between its children.

First we note:
\[ p_4^{S_1} ~=~ m^{\underline 4} / (2m)^4 ~=~ \tfrac{1}{2^4} \big( 1 -
  \tfrac{6}{m} \pm \tfrac{O(1)}{m^2} \big) \]
Indeed, the first key will go to the left child with probability
$\frac{1}{2} = \frac{m}{2m}$. Conditioned on this, the second key
will go to the left child with probability $\frac{m-1}{2m}$, etc.  In
$S_2$, all keys go to the left child with probability a half, so
$p_4^{S_2} = \frac{1}{2}$. Since $P_{S_2} =
\frac{2}{m} \pm O(\frac{1}{m^2})$, we get
\[ p_4^{T_1} ~=~ P_{S_2} \cdot p_4^{S_2} + (1-P_{S_2}) p_4^{S_1}
~=~ \tfrac{1}{2^4} \big( 1 + \tfrac{8}{m} \pm \tfrac{O(1)}{m^2} \big) 
~=~ 2^{-4} +\Theta(1/m). \]
To avoid a rather involved calculation, we will not derive $p_4^{T_2}$
directly, but rather as a function of the 4\th moment. We have $F_4^{S_1} \le
1$, $F_4^{S_3} = \delta^4 = \frac{1}{4} m^2 + O(m^{3/2})$, and $P_{S_3} = 1 - O(\frac{1}{\sqrt{m}})$, so
\[F_4^{T_2} = P_{S_3}F_4^{S_3}+(1-P_{S_3})F_4^{S_1}= \frac{1}{4} m^2 \pm O(m^{3/2}).\]
From the proof of Lemma~\ref{lem:Fk}, we know that $F_4 = (2m)^{\underline 4} p_4 +
f_k(m, p_2, p_3)$ with any distribution. Since $T_2$ is 3-independent, it has the same $p_2$ and $p_3$ as a truly random distribution. Thus, we can compute
$f(m,p_2, p_3)$ using the $p_4$ and $F_4$ values of a truly random
distribution. The 4\th moment of a truly random distribution is:
\[ F_4 = \frac{2m}{2^4} + \binom{4}{2} \frac{(2m)^{\underline 2}}{2^4} =
\frac{24m^2 - 10m}{2^4}. \]
Since $p_4 = 2^{-4}$ in the truly random case, we have: $f(m,p_2, p_3)
= 2^{-4} \big[ (2m)^{\underline 4}  - (24 m^2 - 10 m) \big]$. Now we can return
to $p_4^{T_2}$:
\begin{eqnarray*}
p_4^{T_2} &=& \frac{F_4^{T_2}  + f(m, p_2, p_3)}{(2m)^{\underline 4}}
  = \frac{1}{2^4} \left( \frac{4m^2 \pm O(m^{3/2})}{(2m)^{\underline 4}}
      + 1 - \frac{24 m^2 - 10 m}{(2m)^{\underline 4}} \right)
\\
&=& \frac{1}{2^4} \left( 1 - \frac{20m^2 \pm O(m^{1.5})}{(2m)^{\underline 4}} \right)
 = \frac{1}{2^4} \left( 1 - \frac{20}{m^2} \pm \frac{O(1)}{m^{2.5}} \right)
~=~ 2^{-4} -\Theta(1/m^2).
\end{eqnarray*}
To get $p_4=2^{-4}$ for a given node, we use a strategy $T^*$ that on
the average over a level applies $T_1$ with some probability
$P_{T_1}^*=\Theta(1/m)$; $T_2$ otherwise. However, as stated earlier,
we will often give preference to $T_1$ on the query path, and to $T_2$
elsewhere.

\subsubsection{The Distribution Tree}\label{sec:tree}
We are now ready to describe the mix of strategies used in the binary tree.
On the top $\frac{2}{3} \lg_2 \tsize$ levels, we use the above mentioned 
mix $T^*$ of $T_1$ and $T_2$ yielding a perfect 4-independent distribution 
of the keys at each node.

On the next levels $\ell\geq \frac{2}{3} \lg_2 \tsize$, we will
always use $T_1$ on the query path. For the other nodes, we use $T_1$
with the probability $P_{T_1}^-$ such that if all non-query nodes on
level $\ell$ use the strategy
\begin{description}
\item[$T^-=$] $P_{T_1}^- \times T_1 ~+~ (1-P_{T_1}^-) \times T_2$;
\end{description}
then we get $p_4=2^{-4}$ for an average 4-tuple on level $\ell$. We
note that $P_{T_1}^-$ depends completely on the distribution of 4-tuples at the nodes on level $\ell$ and that $P_{T_1}^-$ has to compensate for the fact
that $T_1$ is used at the query node. We shall prove the existence of
$P_{T_1}^-$ shortly.

Finally, we have a stopping criteria: if at some level
$\ell$, we use the $S_2$ collection on the query path, or if
$\ell+1>\frac{5}{6} \lg_2 \tsize$, then we use a truly random
distribution on all subsequent levels.  We note that the $S_2$
collection could happen already on a top level $\ell\leq \frac23
\lg_2 \tsize$.

\subsubsection{Possibility of Balance}\label{sec:balance}
Consider a level $\ell$ before the stopping criteria has been applied.
We need to argue that the above mentioned probability $P_{T_1}^-$ exists. 
We
will argue that $P_{T_1}^-=0$ implies $p_4<2^{-4}$ while $P_{T_1}^-=1$
implies $p_4>2^{-4}$. Then continuity implies that there exists a 
$P_{T_1}^-\in [0,1]$ yielding $p_4=2^{-4}$.

With $P_{T_1}^-=1$, we use 
strategy $T_1$ for all nodes on the level, and we already know that
$p_4^{T_1}>2^{-4}$. 

Now consider $P_{T_1}^-=0$, that is, we use $T_1$ only at the query node.
Starting with a simplistic calculation, assume that all $2^\ell$ nodes on 
level $\ell$ had exactly $2m=n/2^\ell$ keys, hence the same number of 
4-tuples. Then the average is
\[\frac{p_4^{T_1}+(2^\ell-1)p_4^{T_2}}{2^{\ell}}=
\frac{2^{-4}+\Theta(1/m)-(2^\ell-1)(2^{-4}+\Theta(1/m^2))}{2^{\ell}}<2^{-4}.\]
The inequality follows because $\ell\geq \frac{2}{3} \lg_2 \tsize$
implies $2^{\ell}>n^{2/3}$ while $m<n/2^{\ell}\leq n^{1/3}$. However,
the number of keys at different nodes on level $\ell$ is not expected to
be the same, and we will handle this below.

We want to prove that the average $p_4$ over all 4-tuples on level
$\ell$ is below $2^{-4}$. To simplify calculations, we can add
$p_4^{T_1}-2^{-4}=\Theta(1/m)$ for each 4-tuple using $T_1$ and
$p_4^{T_2}-2^{-4}=-\Theta(1/m^2)$ for each tuple using $T_2$, and show
that the sum is negative. If the query node has $2m$ keys, all using
$T_1$, we thus add $(2m)^{\underline 4}\Theta(1/m)=\Theta(m^3)$. If a
non-query node has $2m$ keys, we subtract $(2m)^{\underline
  4}\Theta(1/m^2)=\Theta(m^2)$.

We now want to bound the number of keys at the level $\ell$ query node. Since the
stopping criteria has not applied, we know that the $S_2$ collection has
not been applied to any of its ancestors.
\begin{lemma}\label{lem:no-collection} If we have never applied the $S_2$ 
collection on the path
to a query node $v$ on level $j\leq
\frac{5}{6} \lg_2 \tsize$, then $v$ has 
$n/2^j \pm 3\sqrt{n/2^j}$ keys.
\end{lemma}
\begin{proof}
On the path to $v$, we have only applied strategies $S_1$ and $S_3$. Hence, if 
an ancestor of $v$ has $2m$ keys, then each child gets
$m\pm (\sqrt {m/2}+1)$ keys. The bound follows by induction
starting with $2m=n$ keys at the root on level $0$.
\end{proof}
Our level $\ell$ query node thus has $\Theta(n/2^\ell)$ keys and contributes
$O((n/2^\ell)^3)$ to the sum.

To lower bound the negative contribution from the non-query nodes on
level $\ell$, we first note that they share all the
$n-O(n/2^\ell)=\Omega(n)$ keys not on the query path. The negative
contribution for a node with $2m$ keys is $\Omega(m^2)$. By convexity,
the total negative contribution is minimized if the keys are evenly
spread among the $2^\ell-1$ non-query nodes, and even less if we
distributed on $2^\ell$ nodes. The total negative contribution is
therefore at least
$2^{\ell}\,\Omega((n/2^{\ell})^2)=\Omega(n^2/2^\ell)$. This dominates
the positive contribution from the query node since $(2^\ell)^2\geq
n^{4/3}=\omega(n)$. Thus we conclude that $p_4<2^{-4}$ when
$P_{T_1}^-=0$.  This completes the proof that we for level $\ell$ can
find a value of $P_{T_1}^-\in[0,1]$ such that $p_4=2^{-4}$, hence
the proof that the distribution tree described in Section
\ref{sec:tree} exists, hashing all keys 4-independently.

\subsubsection{Expected Query Time}
We will now study the expected query cost for our designated query key
$q$. We only consider the cost in the event that the $S_2$ collection
is applied at the query node at some level $\ell\in [\frac{2}{3}
  \lg_2 \tsize,\frac{5}{6} \lg_2 \tsize]$. Assume that this
happened. Then $S_2$ has not been applied previously on the query
path, so the event can only happen once with a given distribution (no
over-counting).  By Lemma \ref{lem:no-collection}, our query node has
$n/2^\ell \pm 3\sqrt{n/2^\ell}$ keys. With probability $1/2$, these
all go to the query child which represents an interval of length
$\tsize/2^{\ell+1}$. Since $n=2t/3$, we conclude that the query child
gets overloaded by almost a factor $4/3$. By Lemma
\ref{lem:overflow-query}, the expected cost of searching $q$ is then
$\Omega(n/2^{\ell})$. This assumed the event that $S_2$ collection was 
applied to the query path on level $\ell$ and not on any level $i<\ell$.

On the query path on every level $i\leq \ell$, we know that the probability of applying $S_2$ 
provided that $S_2$ has not already been applied is $\Theta(1/m)$ where
$m=\Theta(n/2^i)$ by Lemma \ref{lem:no-collection}. The probability
of applying $S_2$ on level $\ell\in [\frac{2}{3} \lg_2 \tsize,\frac{5}{6} \lg_2 \tsize]$ is
therefore $(1-\sum_{i=0}^{\ell-1}O(2^i/n))\Theta(2^\ell/n)=\Theta(2^\ell/n)$,
so the expected search cost from this level is $\Theta(1)$. Since our event
can only happen on one level for a given distribution, we sum this
cost over the $\Omega(\lg n)$ levels in 
$[\frac{2}{3} \lg_2 \tsize,\frac{5}{6} \lg_2 \tsize]$. We conclude that
with our 4-independent scheme, the expected cost of searching the
designated key is $\Omega(\lg n)$. Thus we have proved
\begin{theorem}\label{thm:4-ind-query} 
For any set of $n$ given keys plus a distinct query key,
there exists a 4-independent hashing scheme such that if it
is used to insert the  $n$ keys in a linear probing table of size 
$t=2n$, then the query takes $\Omega(\log n)$ time.
\end{theorem}

\section{Minwise Independence via $k$-Independence}

Recall that a hash function $h$ is \emph{$\eps$-minwise independent} if
for any key set $S$ and distinct query key $q\not\in S$, we have $\Pr_{h}[ h(q) < \min h(S) ]
= \frac{1 \pm \eps}{|S|+1}$.  

Indyk~\cite{indyk01minwise} proved that $O(\lg
\frac{1}{\eps})$-independent hash functions are $\eps$-minwise independent.
His proof is not based on moments but uses another standard tool enabled by $k$-independence: the
inclusion-exclusion principle. Say we want to bound the probability that
at least one of $n$ events $A_0,\ldots,A_{n-1}$ occurs.  Define
$p(k) = \sum_{S\subseteq[n], |S|=k} \Pr\rule{0em}{2ex}\left[\bigcap_{i\in S}
  A_i\right]$. The probability that at least one event occurs
is, by inclusion-exclusion, $\Pr\left[\bigcup_{i\in [n]}A_i\right]=
  p(1) - p(2) + p(3) - p(4) + \dots$, and if $k\leq n$ is odd,
then  $\Pr\left[\bigcup_{i\in [n]}
      A_i\right]\in\left[\sum_{j=1}^{k-1} (-1)^{j-1} p(j),\  \sum_{j=1}^{k} (-1)^{j-1} p(j)\right]$. The gap between the
    bounds is $p(k)$. If the events $A_0,\ldots,A_{n-1}$ are $k$-independent, then
    $p(1),..,p(k)$ have exactly the same values as in the fully independent
    case. Thus, $k$-independence achieves bounds exponentially close
    to those with full independence, whenever probabilities can be
    computed by inclusion-exclusion and $p(k)$ decays exponentially
    in $k$. This turns out to be the case for minwise independence:
    we can express the probability that at least some key in $S$ is
    below $q$ by inclusion-exclusion.

In this paper, we show that, for any $\eps>0$, there exist $\Omega(\lg
\frac{1}{\eps})$-independent hash functions that are no better than
$\eps$-minwise independent. Indyk's ~\cite{indyk01minwise} simple
analysis via inclusion-exclusion is therefore tight: $\eps$-minwise
independence requires $\Omega(\lg \frac{1}{\eps})$-independence.

To prove the result for a given $k$, our goal is to construct a
$k$-independent distribution of hash values for $n$ regular keys and a
distinct query key $q$, such that the probability that $q$ gets the minimal
hash value is $\big( 1+2^{-O(k)} \big)/(n+1)$.

We assume that $k$ is even and divides $n$. Each hash value will be
uniformly distributed in the unit interval $[0,1)$. Discretizing this
  continuous interval does not affect any of the calculations below,
  as long as precision $2\lg n$ or more is used (making the
  probability of a non-unique minimum vanishingly small).

For our construction, we divide the unit interval into $\frac{n}{k}$
subintervals of the form $\big[ i \frac{k}{n}, (i+1) \frac{k}{n} \big)$.
The regular keys are distributed totally randomly between these
subintervals. Each subinterval $I$ gets $k$ regular keys in
expectation. We say that $I$ is \emph{exact} if it gets exactly $k$
regular keys. Whenever $I$ is not exact, the regular keys are placed
totally randomly within it.

The distribution inside an exact interval $I$ is dictated by a parity
parameter $P\in \{0,1\}$. We break $I$ into two equal halves, and
distribute the $k$ keys into these halves randomly, conditioned on the
parity in the first half being $P$. Within its half, each key gets an
independent random value. If $P$ is fixed, this process is 
$(k-1)$-independent. Indeed, one can always deduce the half of a key $x$ based
on knowledge of $k-1$ keys, but the location of $x$ is totally uniform
if we only know about $k-2$ keys. If the parity parameter $P$ is
uniform in $\{0,1\}$ (but possibly dependent among exact intervals), the
overall distribution is still $k$-independent.

The query is generated uniformly and independent of the
distribution of regular keys into intervals. For each exact interval
$I$, if the query is inside it, we set its parity parameter $P_I =
0$. If $I$ is exact but the query is outside it, we toss a biased coin
to determine the parity, with $\Pr[P_I = 0] = (\frac{1}{2} -
\frac{k}{n}) / (1 - \frac{k}{n})$. Any fixed exact interval receives
the query with probability $\frac{k}{n}$, so overall the distribution
of $P_I$ is uniform. It is only via these parity parameters that 
the query effects the distribution of the regular keys within the intervals.

We claim that the overall process is $k$-independent. Uniformity of
$P_I$ implies that the distribution of regular keys is
$k$-independent. In the case of $q$ and $k-1$ regular keys, we also
have full independence, since the distribution in an interval is
$(k-1)$-independent even conditioned on $P$.

It remains to calculate the probability of $q$ being the minimum under
this distribution. First we assume that the query landed in an exact
interval $I$, and calculate $p_{\min}$, the probability that $q$ takes
the minimum value within $I$. Define the random variable $X$ as the
number of regular keys in the first half. By our process, $X$ is
always even.

If $X=x > 0$, $q$ is the minimum only if it lands in the first half
(probability $\frac{1}{2}$) and is smaller than the $x$ keys already
there (probability $\frac{1}{x+1}$). If $X=0$, $q$ is the minimum
either if it lands in the first half (probability $\frac{1}{2}$), or
if it lands in the second half, but is smaller than everybody there
(probability $\frac{1}{2 (k+1)}$). Thus,
\[ p_{\min} = 
  \Pr[X=0] \cdot \big( \tfrac{1}{2} + \tfrac{1}{2 (k+1)} \big)
+ \sum_{x=2,4,\twodots, k} \Pr[X=x] \cdot \tfrac{1}{2(x+1)} \]

To compute $\Pr[X=x]$, we can think of the distribution into halves as
a two step process: first $k-1$ keys are distributed randomly; then,
the last key is placed to make the parity of the first half
even. Thus, $X=x$ if either $x$ or $x-1$ of the first $k-1$ keys
landed in the first half. In other words:
\[ \Pr[X=x] = \tbinom{k-1}{x} / 2^{k-1} + \tbinom{k-1}{x-1} / 2^{k-1}
            = \tbinom{k}{x} / 2^{k-1}\] 

No keys are placed in the first half iff none of the first $k-1$ keys
land there; thus $\Pr[X=0] = 1/2^{k-1}$. We obtain:
\[ p_{\min} = \frac{1}{2^k (k+1)} + \frac{1}{2^k}
    \sum_{x=0,2,\twodots, k} \frac{1}{x+1} \binom{k}{x} \]
But $\frac{1}{x+1} \binom{k}{x} = \frac{1}{k+1} \binom{k+1}{x+1}$.
Since $k+1$ is odd, the sum over all odd binomial coefficients is
exactly $2^{k+1}/2$ (it is equal to the sum over even binomial
coefficients, and half the total). Thus, $p_{\min} = \frac{1}{2^k
  (k+1)} + \frac{1}{k+1}$, i.e.~$q$ is the minimum with a probability
that is too large by a factor of $1+2^{-k}$.

We are now almost done. For $q$ to be the minimum of all keys, it has
to be in the minimum non-empty interval. If this interval is exact,
our distribution increases the chance that $q$ is minimum by a factor
$1+2^{-k}$; otherwise, our distribution is completely random in the
interval, so $q$ is minimum with its fair probability. Let $Z$ be the
number of regular keys in $q$'s interval, and let $\evt$ be the event
that $q$'s interval is the minimum non-empty interval. If the
distribution were truly random, then $q$ would be minimum with
probability:
\[ \frac{1}{n+1} = \sum_z \Pr[Z=z] \cdot \Pr[\evt \mid Z=z] \cdot
\frac{1}{z+1} \]
In our tweaked distribution, $q$ is minimum with probability:
\begin{eqnarray*}
& & \sum_{z\ne k} \Pr[Z=z] \cdot \Pr[\evt \mid Z=z] \cdot \frac{1}{z+1} 
    + \Pr[Z=k] \cdot \Pr[\evt \mid Z=k] \cdot \frac{1+2^{-k}}{k+1}
\\
&=& \frac{1}{n+1} ~+~ \Pr[Z=k]\cdot \Pr[\evt \mid Z=k] \cdot \frac{2^{-k}}{k+1}
\end{eqnarray*}

But $Z$ is a binomial distribution with $n$ trials and mean $k$; thus
$\Pr[Z=k] = \Omega(1/\sqrt{k})$. Furthermore, $\Pr[\evt \mid Z=k] \ge
\frac{k}{n}$, since $q$'s interval is the very first with probability
$\frac{k}{n}$ (and there is also a nonzero chance that it is not the
first, but all interval before are empty). Thus, the probability is
off by an additive term $\frac{\Omega(2^{-k} / \sqrt{k})}{n}$. This
translates into a multiplicative factor of $1 + 2^{-O(k)}$.  Thus we have proved
\begin{theorem}\label{thm:k-ind-min} 
For any set $S$ of $n$ given keys plus a query key $q\not\in S$,
there exists a $k$-independent scheme $h$ such that 
$\Pr_h[h(q)<\min h(S)]=(1+1/2^{O(h)})/(n+1)$.
\end{theorem}

\section{Multiply-Shift hashing}\label{sec:ms}
\newcommand{\calU}{\mathcal{U}} \newcommand{\calO}{\mathcal{O}} 

We will now show that the simplest and fastest known universal
\cite{dietzfel97closest} and 2-independent \cite{dietzfel96universal}
hashing schemes have bad expected performance when used for linear
probing and minwise hashing on some of the most common structured
data; namely a set of consecutive numbers. This is a nice contrast to
the result of Mitzenmacher and Vadhan \cite{mitzenmacher08hash} that
any 2-independent hashing scheme works if the input data has enough
entropy.

\subsection{Linear probing}\label{sec:ms-lin} Our result is inspired by
negative experimental findings from \cite{thorup12kwise}.  The
essential form of the schemes considered have the following basic
form: we want to hash $\ell_{in}$-bit keys into $\ell_{out}$-bit
indices. Here $\ell_{in}\geq\ell_{out}$, and the indices are used for
the linear probing array. For the typical case of a half full table,
we have $2^{\ell_{out}}=t\approx 2 n$. In particular, $t>n$.

Depending on details of the scheme, for some $\ell\geq
\ell_{in},\ell_{out}$, we pick a random multiplier $a\in [2^\ell]$,
and compute
\begin{equation}\label{eq:mult-shift}
h_a(x)=\lfloor (ax\bmod 2^\ell)/2^{\ell-\ell_{out}}\rfloor.
\end{equation}
We are going to show that if we use this
scheme for a linear probing table of size $t=2^{\ell_{out}}=2n$, and if we try to insert the keys
in $[n]=\{0,\ldots,n-1\}$, then the expected average insertion time is $\Omega(\log n)$.

We refer to the scheme in \req{eq:mult-shift} as the {\em basic multiply-shift scheme}.
The mod-operation is easy, as we just have to discard overflowing
bits. If $\ell\in\{8,16,32,64\}$, this is
done automatically in a programming language like C \cite{KR88}.
 The division with rounding is just a right shift by
$s=\ell-\ell_{out}$, so in C we get the simple code \texttt{(}$a$\texttt{*}$x$\texttt{)>>}$s$
and the cost is dominated by a single multiplication.  For the plain
universal hashing in \cite{dietzfel97closest}, it suffices that $\ell\geq\ell_{in}$ but then the multiplier $a$
should be odd. For 2-independent hashing as in \cite{dietzfel96universal}, 
we need $\ell\geq
\ell_{in}+\ell_{out}-1$. Also we need to add a random number $b$, but
as we shall discuss in the end, these details have no essential impact
on our analysis. However, our lower bounds for linear
probing do assume that the last shift takes out at least one bit, hence that 
\begin{equation}\label{eq:ell-out}
\ell>\ell_{out}.
\end{equation}
It is instructive to compare \req{eq:mult-shift} with the
corresponding classic scheme $((ax+b) \bmod p)\bmod 2^{\ell_{out}}$
for some large enough prime $p$. For this classic scheme,
\cite{pagh07linprobe} already proved an $\Omega(\log n)$ lower bound on
the average insertion time but with a different bad instance. The first
mod-operation in the classic scheme is with a prime instead of the
power of two \req{eq:mult-shift}. The second mod-operation in the classic scheme limits the
range to $\ell_{out}$-bit integers by saving the $\ell_{out}$ least
significant bits whereas the corresponding division in
\req{eq:mult-shift} saves the $\ell_{out}$ most significant
bits. These differences both lead to a quite different mathematical
analysis.

As mentioned, our basic bad example will be where the keys form the 
interval $[n]$. However, the problem will not
go away if this interval is shifted or not totally full, or replaced
by an arithmetic progression.

When analyzing the scheme, it is convenient
to view both the multiplier and the hash value before
the division as fractions in the 
unit interval $[0,1)$, defining $a^\downarrow=a/2^\ell$, and
\[h^\downarrow_a(x)=(ax\bmod 2^\ell)/2^{\ell}=a^\downarrow x\bmod 1.\]
Then $h_a(x)=\lfloor h^\downarrow_a(x)2^{\ell_{out}}\rfloor$. We think
of the unit interval as circular, and for any
$x\in[0,1)$, we define 
\[\|x\|=\min \{x\bmod 1, -x\bmod 1\}.\]
This is the distance from 0 in the circular unit interval.
\begin{lemma}\label{lem:avg-cost}
Let the multiplier $a$ be given and suppose for some $x\in \{1,\ldots,n-1\}$ 
that $\| h^\downarrow_a(x)\|\leq 1/(2t)$. Then, when
we use $h_a$ to hash $[n]$ into a linear probing table,
the average cost per key is $\Omega(n/x)$.
\end{lemma}
\begin{proof}
The case studied is illustrated in Figure~\ref{fig:5cycle}.
\begin{figure}
\centerline{\includegraphics[width=0.3\textwidth]{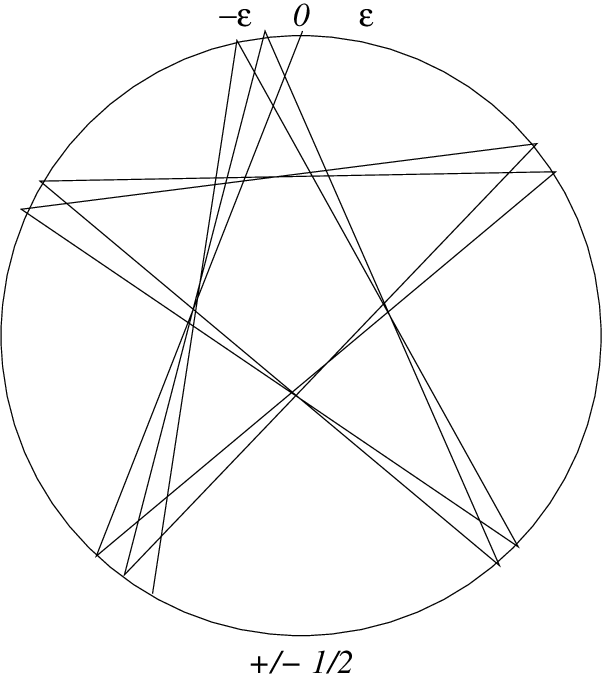}}
\caption{Case where $\|h^\downarrow_a(5)\|\leq \eps$.}\label{fig:5cycle}
\end{figure}
We can assume that $n/x\geq 8$ since the cost of inserting a key is
always at least a constant.  For each $k\in[x]$, consider the set
$[n]^x_k=\{y\in[n]\;|\;y\bmod x=k\}$.  The keys in $[n]^x_k$ are only $1/(2t)$
apart since for every $y$, $h^\downarrow_a(y+x)-h^\downarrow_a(y)=h^\downarrow_a(x)$. Therefore
the $q\geq \lfloor n/x
\rfloor\geq 8$ keys from $[n]^x_k$ map to an interval of length
$(q-1)/(2t)$, which means that $h_a$ distributes $[n]^x_k$ on at most
$\lceil q/2\rceil+1<3q/4$ consecutive array locations.  Linear probing
will have to spread $[n]^x_k$ on $q$ locations, so on the average, the
keys in $[n]^x_k$ get a displacement of $\Omega(q)=\Omega(n/x)$. 
This analysis applies to every equivalence class modulo
$x$, so we get an average insertion cost of $\Omega(n/x)$ over all
the keys.  The above average costs only measures the interaction among
keys from the same equivalence class modulo $x$. If the ranges of hash
values from different classes overlap, the cost will be bigger.
\end{proof}
Note that $\|h^\downarrow_a(x)\|\leq 1/(2t)$ implies that
$h^\downarrow_a(x)$ is contained in an interval of size $1/t$ around 0.
From the universality arguments of \cite{dietzfel97closest,dietzfel96universal} 
we know that
the probability of this event is roughly $1/t$ (we shall return with an
exact statement and proof later).
We would like to conclude
that the expected average cost is $\sum_{x=1}^n \Omega(n/x)/t=\Omega(\lg n)$.
The answer is correct, but the calculation cheats in the sense
that for a single multiplier $a$, we may have many different $x$ such that $\|h^\downarrow_a(x)\|\leq 1/(2t)$,
and the associated costs should not all be added up.

To get a proper lower bound, for any given multiplier $a$, we let
$\mu_a$ denote the minimal positive value such that $\|h^\downarrow_a(\mu_a)\|\leq
1/(2t)$.  We note that there cannot be any $x<y<\mu_a$ 
at distance at most $1/(2t)$, for then we would have
$\|h^\downarrow_a(y-x)\|=\|h^\downarrow_a(y)-h^\downarrow_a(x)\|\leq 1/(2t)$.

If $\mu_a<n$, then by Lemma \ref{lem:avg-cost}, the average
insertion cost over keys is $\Omega(n/\mu_a)$. Therefore, if $a$ is random over some
probability distribution (to be played with as we go along), the 
expected (over $a$) average (over keys) insertion cost is lower bounded by
\begin{equation}\label{eq:cost}
\Omega\left(\sum_{x=1}^n \Pr_a[\mu_a=x]n/x\right).
\end{equation}
\begin{lemma}\label{lem:bad} For a given multiplier $a$, consider any $x<n$ 
such that 
$\|h^\downarrow_a(x)\|\leq 1/(2t)$. Then $x\neq\mu_a$ if and only if
for some prime factor $p$ of $x$, $\|h^\downarrow_a(x/p)\|
\leq 1/(2pt)$. 
\end{lemma}
\begin{proof}
The ``if'' part is trivial. By minimality of $\mu_a$, we have
$x>\mu_a$. 

Since $\|h^\downarrow_a(\mu_a)\|\leq 1/(2t)$, for
any integer $i<t$, we have $\|h^\downarrow_a(i\mu_a)\|=i\|h^\downarrow_a(\mu_a)\|$.
Suppose now that $x=j\mu_a$. Then $1<j\leq x<n<t$, so
for any $i\leq j$, we have $\|h^\downarrow_a(i\mu_a)\|\leq 
\|h^\downarrow_a(x)\|\leq  1/(2t)$. We can
therefore take any prime factor $p$ of $j$, and conclude that
$\|h^\downarrow_a(x/p)\|\leq \|h^\downarrow_a(x)\|\leq 1/(2t)$. Since $p$ is also
a prime factor of $x$, this proves the lemma if $x$ is a multiple of $\mu_a$.

To complete the proof we will argue that $x$ has to be a multiple of $\mu_a$.
Consider any $y$ such that $\|h^\downarrow_a(y)\|\leq 1/(2t)$ where
$y$ is not a multiple of $\mu_a$. Then $h^\downarrow_a$ maps
$\{0,\ldots,y+\mu_a-1\}$ to points in the cyclic unit interval that
are at most $1/(2t)$ apart (c.f.,  Figure~\ref{fig:5cycle}). It follows
that $y\geq 2t-\mu_a$. However, we have $\mu_a<x<n<t$, which
implies that $x< 2t-\mu_a\leq y$. It follows that $x$ has to be a multiple
of $\mu_a$.
\end{proof}
To illustrate the basic accounting idea,
assume for simplicity that we have a perfect distribution $\calU$ on $a$ that
for any fixed $x>0$ distributes $h^\downarrow_a(x)$ uniformly in the unit interval. Then for
any $x$ and $\eps<1/2$,
\begin{equation}\label{eq:fake}
\Pr_{a\leftarrow\calU}[\|h^\downarrow_a(x)\|\leq \eps]=2\eps.
\end{equation}
Then by Lemma \ref{lem:bad},
\begin{eqnarray}
\Pr_{a\leftarrow\calU}[\mu_a=x]
&\geq&\Pr_{a\leftarrow\calU}[\| h^\downarrow_a(x)\|\leq 1/(2t)]-\sum_{p\ {\rm prime\ factor\ of}\ x} \Pr_{a\leftarrow\calU}[\| h^\downarrow_a(x/p)\|\leq 1/(2pt)]\nonumber\\
&=&1/t-\sum_{p\ {\rm prime\ factor\ of}\ x} 1/(pt)\nonumber\\
&=& \left(1-\sum_{p\ {\rm prime\ factor\ of}\ x} 1/p\right)/t\label{eq:prop-min}
\end{eqnarray}
We note that the lower-bound \eqref{eq:prop-min} may be 
negative since there are values of $x$ for which 
$\sum_{p\ {\rm prime\ factor\ of}\ x} 1/p=\Theta(\lg\lg x)$. Nevertheless \eqref{eq:prop-min}
suffices with an appropriate reordering of terms. 
From \eqref{eq:cost} we get that the expected average insertion cost
is lower bounded within a constant factor by:
\begin{eqnarray*}
\sum_{x=1}^{n} \Pr_{a\leftarrow\calU}[\mu_a=x]n/x
&\geq &\sum_{x=1}^{n} \left(1-\sum_{{\rm prime\ factor}\ p{\rm\ of}\ x} 1/p\right)n/(xt)\\
&>&
\sum_{x=1}^{n} \left(1-\sum_{{\rm prime}\ p=2,3,5,..} 1/p^2\right)n/(xt)
\end{eqnarray*}
Above we simply moved terms of the form $-n/(xmp)$ where $p$ is a prime 
factor of $x$ to $x'=x/p$ in the form $-n/(x'mp^2)$. Conservatively,
we include $-n/(x'mp^2)$ for all primes $p$ even if $px'>n$.
Since $\sum_{{\rm prime}\ p=2,3,5,..} 1/p^2<0.453$, we get an expected average insertion cost 
of
\begin{eqnarray*}
\Omega\left(\sum_{x=1}^{n} \Pr_{a\leftarrow\calU}[\mu_a=x]n/x\right)
&=&\Omega\left(\sum_{x=1}^{n} 0.547n/(xt)\right)\\
&=&\Omega((n/t)\lg n).
\end{eqnarray*}
We would now be done if we had the perfect distribution $\calU$ on $a$ so
that the equality \eqref{eq:fake} was satisfied. Instead we will
use the weaker statements of the following lemma:
\begin{lemma}\label{lem:eq:approx}
Let $\calO$ be the uniform
distribution on odd $\ell$-bit numbers. For any odd $x<n$
and  $\eps< 1/2$,
\begin{equation}\label{eq:hit-eps}
\Pr_{a\leftarrow\calO}[\|h^\downarrow_a(x)\|\leq \eps]\leq 4\eps
\end{equation}
However, if $\eps$ is an integer multiple of $1/2^\ell$, then
\begin{equation}\label{eq:hit-int}
\Pr_{a\leftarrow\calO}[\|h^\downarrow_a(x)\|\leq \eps]\geq 2\eps.
\end{equation}
\end{lemma}
\begin{proof} 
  When $x$ is odd and $a$ is a uniformly distributed odd $\ell$-bit
  number, then $ax\bmod 2^\ell$ is uniformly distributed odd
  $\ell$-bit number. To get $h^\downarrow_a(x)$, we divide by $2^\ell$, 
  and then we have a uniform distribution
  on the $2^{\ell-1}$ odd multiples of $1/2^\ell$. Therefore
  $\Pr_{a\leftarrow\calO}[\|h^\downarrow_a(x)\|\leq \eps]/\eps$ is
maximized when $\eps=1/2^\ell$, in which case
$\Pr_{a\leftarrow\calO}[\|h^\downarrow_a(x)\|\leq 1/2^\ell]=2/2^{\ell-1}=4\,2^\ell$, matching
the upper bound in \req{eq:hit-eps}.

When $\eps=i/2^\ell$ for some integer $i$, we minimize   $\Pr_{a\leftarrow\calO}[\|h^\downarrow_a(x)\|\leq \eps]/\eps$ when $i$ is even, in which case
we get $\Pr_{a\leftarrow\calO}[\|h^\downarrow_a(x)\|\leq i/2^\ell]=i/2^{\ell-1}=2i/2^\ell$, matching
the lower bound in \req{eq:hit-int}.
\end{proof}
We are now ready to prove our lower bound for the performance of
linear probing with the basic multiply-shift scheme with an odd multiplier.
\begin{theorem}\label{thm:bad-multiply-shift-lin}
  Suppose $\ell_{out}<\ell$
  and that the multiplier $a$ is a uniformly
  distributed odd $\ell$-bit number.  If we use $h_a$ to insert $[n]$
  in a linear probing table, then the expected average insertion cost is
  $\Omega(\lg n)$.
\end{theorem}
\begin{proof}
By assumption $1/(2t)=1/2^{\ell_{out}+1}$ is a multiple of $1/2^\ell$, 
so for odd $x<n$,
\eqref{eq:hit-int} implies
\begin{equation}\label{eq:hit-m}
\Pr_{a\leftarrow\calO}[\|h^\downarrow_a(x)\|\leq 1/(2t)]\geq 1/t.
\end{equation}
By Lemma \ref{lem:bad} combined with \eqref{eq:hit-eps} and \eqref{eq:hit-m},
we get for any given odd $x$ that 
\begin{eqnarray}
\Pr_{a\leftarrow\calO}[\mu_a=x]
&\geq&\Pr_{a\leftarrow\calO}[\| h^\downarrow_a(x)\|\leq 1/(2t)]-\sum_{p\ {\rm prime\ factor\ of}\ x} \Pr_{a\leftarrow\calO}[\| h^\downarrow_a(x/p)\|\leq 1/(2pt)]\nonumber\\
&\geq&1/t-\sum_{p\ {\rm prime\ factor\ of}\ x} 2/(pt)\label{eq:mu}
\end{eqnarray}
From \eqref{eq:cost} we get that the expected average insertion cost
is lower bounded within a constant factor by:
\begin{eqnarray}
\sum_{\textnormal{odd }x=1}^{n} \Pr_{a\leftarrow\calO}[\mu_a=x]n/x
&\geq &\sum_{\textnormal{odd }x=1}^{n} \left(1-2\sum_{{\rm prime\ factor}\ p{\rm\ of}\ x} 1/p\right)n/(xt)\nonumber\\
&>&
\sum_{\textnormal{odd }x=1}^{n} \left(1-2\sum_{{\rm prime}\ p=3,5,..} 1/p^2\right)n/(xt)\nonumber\\
&>&
\sum_{\textnormal{odd }x=1}^{n} 0.594\,n/(xt)\nonumber\\
&>& 0.298 (n/t) H_n.\label{eq:calc}
\end{eqnarray}
Above we again moved terms of the form $-n/(xmp)$ where $p$ is a prime 
factor of $x$ to $x'=x/p$ in the form $-n/(x'mp^2)$. Since
$x$ is odd, we only have to consider odd primes factors $p$, and then we used
that $\sum_{{\rm prime}\ p=3,5,..} 1/p^2<0.203$. This completes the
proof of Theorem \ref{thm:bad-multiply-shift-lin}.
\end{proof}
We note that the plain universal hashing from \cite{dietzfel97closest} also
assumes an odd multiplier, so Theorem~\ref{thm:bad-multiply-shift-lin}
applies directly if $\ell_{out}<\ell$. The condition $\ell_{out}<\ell$ 
is, in fact, necessary for
bad performance. If $\ell_{out}=\ell$, then $h_a$ is a permutation for
any odd $a$, and then linear probing works perfectly.

For the 2-independent hashing in \cite{dietzfel96universal} there are two
differences. One is that the multiplier may also be even, but restricting
it to be odd can at most double the cost. The other difference is that
we add an additional 
$\ell$-bit 
parameter $b$, yielding a scheme of the form:
\[h_{a,b}(x)=\lfloor ((ax+b)\bmod 2^{\ell})/2^{\ell-\ell_{out}}\rfloor.\]
The only effect of $b$ is a cyclic shift of the double full buckets, and
this has no effect on the linear probing cost. For the 2-independent hashing,
we have $\ell\geq \ell_{in}+\ell_{out}-1$, so $\ell<\ell_{out}$ if $\ell_{in}>1$.
Hence again we have an expected average linear probing cost of 
$\Omega((n/t)\lg n)$.

Finally, we sketch some variations of our bad input. Currently, we just
considered the set $[n]$ of input keys, but it makes no essential
difference if instead for some integer constants $\alpha$ and $\beta$,
we consider the arithmetic sequence $\alpha [n]+\beta=\{\alpha i
+\beta\,|\,i\in [n]\}$. The $\beta$ just adds a cyclic shift like
the $b$ in 2-independent hashing. If $\alpha$ is odd, then it 
is absorbed in the random multiplier $a$. What we get now is that if for
some $x\in[n]$, we have $\| h^\downarrow_a(\alpha x)\|\leq 1/(2t)$,
then again we get an average cost $\Omega(n/x)$. A consequence
is that no odd multiplier $a$ is universally safe because there
always exists an inverse $\alpha$ (with $a\alpha\bmod 2^\ell=1$) 
leading to a linear cost if $h_a$ is used to insert $\alpha [n]+\beta$.
It not hard to also construct bad examples for even $\alpha$. If
$\alpha$ is an odd multiple of $2^i$, we just have to strengthen
the condition $\ell_{out}<\ell$ to $\ell_{out}<\ell-i$ to get
the expected average insertion cost of $\Omega((n/t)\lg n)$. This kind
of arithmetic sequences could be a true practical problem. For example,
in some denial-of-service attacks, one often just change some bits in the
middle of a header key, and this gives an arithmetic sequence.

Another more practical concern is if the input set $X$ is an $\eps$-fraction
of $[n]$. As long as $\eps>2/3$, the above proof works almost unchanged.
For smaller $\eps$, our bad case is if $\| h^\downarrow_a(x)\|\leq \eps/(2t)$. In that case, for each $k\in[x]$, the
$q=\lfloor n/x\rfloor$ potential keys $y$ from $[n]$ with $y\bmod x=k$ 
would map to an interval of length $\eps(q-1)/(2t)$. This
means that $h_a$ spreads these potential keys 
on at most $\lceil \eps q/2\rceil+1$ consecutive array locations. 
A $\eps$-fraction of these keys are real, so on the average, 
these intervals become double full, leading to an average cost of 
$\Omega(\eps n/x)$. Strengthening $\ell_{out}<\ell$ 
to $\eps\geq 2^{\ell_{out}-\ell}$, we essentially get that all
probabilities are reduced by $\eps$. Thus we end
with a cost of 
$\Omega(\eps^2 (n/t)\lg n)=\Omega(\eps (|X|/t)\lg n)$.

\subsection{Minwise Independence}
We will now demonstrate the lack of minwise independence with 
a hashing scheme of the form
\[h_{a,b}(x)=(ax+b)\bmod 2^{\ell}.\] 
Here $\ell$ is an integer and $a$, $b$, and $x$ are all $\ell$-bit
integers.  Restricting the random parameter $a$ to be odd, it is relatively prime to
$2^{\ell}$, and then $h_{a,b}$ is a permutation. We also note that
here, for minwise hashing, we need the random parameter $b$; for with
$b=0$, we always have $h_{a,0}(0)=0$, which is the unique smallest hash value.
We are going to prove
that this kind of scheme is $\Omega(\log n)$-minwise independent. More precisely,
\begin{theorem}\label{thm:bad-multiply-shift-min}
  Suppose the multiplier $a$ is a uniformly
  distributed odd $\ell$-bit number and that $b$ is uniformly
  distributed $\ell$-bit number. Let $n\in [2^{\ell-1}]$ and 
  $n\leq u\in[2^\ell]$. Then for a uniformly distributed query key in
$[u]\setminus [n]$, we have $\Pr[h_{a,b}(q)<\min h_{a,b}([n])]=\Omega((\log n)/n)$.
\end{theorem}
Before proving the theorem, we discuss its implications. 
First note that for $u=n+1$, the query key is fixed as $q=n$. In this case, the
same lower bound is proved in \cite{broder98minwise} when
the hash function is computed modulo a prime instead of a power of two.
Multiplication modulo a power of two is much faster, and the mathematical analysis is different.

The interesting point in $u\gg n$ is that it corresponds to the case
of a random outlier $q$ versus the dense set $[n]$. By
Theorem \ref{thm:bad-multiply-shift-min}, such an outlier
is disproportionally likely to get the smallest hash value.

Having universe size $u\ll 2^{\ell}$ means that even if we try using
far more random bits $\ell$ than required for the key universe $[u]$,
then this does not resolve the problem that a uniform query $q$ is
disproportionally likely to get the smallest hash value.

Theorem \ref{thm:bad-multiply-shift-min} implies bad minwise performance for many
variants of the scheme. First, if we remove the restriction that $a$
is odd, it can at most halve the probability that $h_{a,b}(q)<\min
h_{a,b}([n])$ so we would still have $\Pr[h_{a,b}(q)<\min
  h_{a,b}([n])]=\Omega((\log n)/n)$. Moreover, this could introduce
collisions, and then we are more concerned with the event
$h_{a,b}(q)\leq \min h_{a,b}([n])$ since ties might be broken
adversarially. Also, as in Section \ref{sec:ms-lin}, if
we only want an $\ell_{out}<\ell$ bits in the hash value, we
can shift out the $\ell-\ell_{out}$ least significant bits, but
this can only increase the chance that $h_{a,b}(q)\leq \min h_{a,b}([n])$.

\begin{proof}[ of Theorem {\ref{thm:bad-multiply-shift-min}}]
As in Section \ref{sec:ms-lin}, it is
convenient to divide $\ell$-bit numbers by $2^\ell$ to get fractions 
in the cyclic unit
interval.  We define $a^\downarrow=a/2^\ell$, $b^\downarrow=b/2^\ell$,
and  
\[h^\downarrow_{a,b}(x)=h_{a,b}(x)/2^\ell=(a^\downarrow x+b^\downarrow)\bmod 1.\]
We note that $h^\downarrow_{a,0}=h^\downarrow_{a}$ from Section
\ref{sec:ms-lin}. In our analysis, we are first going to pick $a$, and
study how $h^\downarrow_{a}$ maps $[n]$ and the random query $q$. 
This analysis will reuse many of the elements from 
Section~\ref{sec:ms-lin} illustrated in Figure \ref{fig:5cycle}.
Later, we will pick the random $b$, which corresponds to a random cyclic rotation by $b^\downarrow$, so that $0$ ends up in what was
position $1-b^\downarrow$ in the image under $h^\downarrow_{a}$

Let $t$ be the smallest power of two not smaller than $n$. Then
$n\leq t\leq 2^\ell/2$. As in Section~\ref{sec:ms-lin},
for any $a$, we define $\mu_a>0$ to be
the smallest number such that $\|h^\downarrow_{a}(\mu_a)\|\leq 1/(2t)$. We are only interested
in the case where $\mu_a<n/4$.  
\drop{
We are essentially want to prove
\begin{equation*}
\Pr[h_{a,b}(q)<\min h_{a,b}([n])]=\Omega(1/\mu_a)
\end{equation*}
Theorem \ref{thm:bad-multiply-shift-min} will then follow from
calculations similar to those in \req{eq:calc}.}

In our cyclic unit interval, we generally view values in $(0,1/2)$ as 
positive and values in $(1/2,1)$ as negative. Also, a value is
between two other values, it is on the short side between them. Positive is
clockwise.

For simplicity, we assume 
that $h^\downarrow_{a}(\mu_a)$ is positive and let 
$\eps_a=h^\downarrow_{a}(\mu_a)$. 
We now claim that the points in $h^\downarrow_{a}([\mu_a])$ are almost
equidistant. More precisely,
\begin{lemma}\label{lem:eq-dist} Considering the points $h^\downarrow_{a}([\mu_a])$ in the
cyclic unit interval, the distance between neighbors is $1/\mu_a\pm\eps_a$.
\end{lemma}
\begin{proof} Let $a'=a^\downarrow-\eps_a/\mu_a$.
Then $a'\mu_a\mod 1=0$. We claim that the $\mu_a$ points in
$a'[\mu_a]\mod 1$ have distance exactly $1/\mu_a$ between
neighbors. Assume for a contradiction, that this is not the case. Then
there should to be some distinct $x,y\in[\mu_a]$ with
$(h^\downarrow_{a}(y)-h^\downarrow_{a}(x))\bmod 1=\Delta<1/\mu_a$. Let
$z=(y-x)\bmod \mu_a$. Then $a'z\bmod 1=\Delta$. Therefore, for every
$i=0,...,\mu_a$, we have $a' iz\bmod 1=i\Delta<1$, and these are
$\mu_a+1$ distinct values. However, $a' iz\bmod 1=a' (iz\bmod
\mu_a)\bmod 1$, so there can only be $\mu_a$ distinct values, hence
the desired contradiction. 

We now know that the points in 
$a'[\mu_a]\mod 1$ have distance exactly $1/\mu_a$ between neighbors, and
for every $x\in [\mu_a]$, we have $h^{\downarrow}(x)=a'x+\eps x/\mu_a \mod 1$
where $\eps x/\mu_a<\eps$. Hence follows that 
distance between any neighbors in $h^\downarrow_{a}([\mu_a])$ is $1/\mu_a\pm\eps_a$.
\end{proof}
Points from $h^\downarrow_a ([\mu_a])$ divide the cyclic unit
interval into $\mu_a$ ``slices''. By Lemma \ref{lem:eq-dist}, each slice
is of length at least $1/\mu_a-\eps_a$. 
Consider some
$k\in [\mu_a]$. The keys $x=k,k+\mu_a,k+2\mu_a,\ldots$, map to
$h^\downarrow_a(k),h^\downarrow_a(k)+\eps_a,h^\downarrow_a(k)+2\eps_a,...$.
We call this the ``thread'' from $h^\downarrow_a(k)$. Thus, for $x\geq \mu_a$, 
$h^\downarrow_a(x)$ is the successor at distance $\eps_a$ from $h^\downarrow_a(x-\mu_a)$ in the thread from $h^\downarrow_a(x\bmod \mu_a)$.

We now consider the image by $h^\downarrow_a$ of our set $[n]$.
For each $k\in[\mu_a]$, the set $[n]^{\mu_a}_k=\{x\in[n]\;|\;x\bmod {\mu_a}=k\}$ has
$d\leq \ceil{n/\mu_a}$ keys that fall in the interval $[h^\downarrow_a(k),
(h^\downarrow_a(k)+d\eps_a)]$ of length $(d-1)\eps_a<(n/\mu_a)\eps_a\leq
1/(2\mu_a)$. We call this the ``filled'' part of the slice, the rest
is ``empty''. The empty part of any slice is bigger than $(1/\mu_a-\eps_a)-1/(2\mu_a)=1/(2\mu_a)-\eps_a$.

We are will study the ``good'' event that $h^\downarrow_a(q)$ and
$1-b^\downarrow$ land strictly inside the empty part of the same slice, for then
with $h_{a,b}$, there is no key from $[n]$ that hash between $0$ and
hash of the query key. If in addition $1-b^\downarrow$ is before
$h^\downarrow_{a}(q)$, then $h_{a,b}(q)<\min
h_{a,b}([n])$. Otherwise, we shall refer to a symmetric case.
\begin{lemma}\label{lem:b-empty}  With $\mu_a\leq n/4$, 
the probability that $1-b^\downarrow$ hash to the empty part of a given slice is at least $1/(4\mu_a)$.
\end{lemma}
\begin{proof}
We know from above that the empty part of any slice is bigger than  
$1/(2\mu_a)-\eps_a$. However, both $1-b^\downarrow$ and the end-points
of the empty interval fall on multiples of $1/2^\ell$, and we want
$1-b^\downarrow$ to fall strictly between the end-points. Since
$1-b^\downarrow$ is uniformly distributed on multiples of $1/2^\ell$,
we get that it falls strictly inside with probability at
least $1/(2\mu_a)-\eps_a-1/2^\ell$. 

Our parameters are chosen such that $\eps_a\leq 1/(2t)\leq 1/2^\ell$,
$n\leq t$, and $\mu_a\leq n/4$, so $1/(2\mu_a)-\eps_a-1/2^\ell\geq 1/(4\mu_a)$. 
\end{proof}

\begin{lemma}\label{lem:q-empty} For any value $u\in (n,2^\ell)$, at least half
the keys in $[u]\setminus[n]$ hash to the empty part of some slice.
\end{lemma}
\begin{proof}
We now consider the potential values of the query key $q=n,...,2^\ell-1$. First, let
$\mu_a^*\in[n,2^\ell-1)$ be the smallest value such that
$\|h_a(\mu_a^*)\|<\eps_a$.  For now we assume that
such a key $\mu_a^*$ exists. We note that $h_a(\mu_a^*)$ must be
negative, for if it was positive, then
$h_a(\mu^*_a-\mu_a)=h_a(\mu^*_a)-\eps_a$, would also satisfy the
condition. We also note that $h_a(\mu_a^*)$ cannot be zero since
$h^\downarrow_a$ is a permutation. Thus we must have
Thus $h_a(\mu^*_a)\in (2^\ell-\eps_a,2^\ell]$.

By definition, all points in $h_a([\mu^*_a])$ are at least $\eps_a$ apart,
so $\mu^*_a\geq 2^\ell/\eps_a\leq 2n$. On the other hand,
$h_a([\mu^*_a, \mu^*_a+\mu_a-1])$ provides a predecessor 
at distance $\eps^*_a<\eps_a$ to every point in $h_a([\mu_a])$, so in 
$h_a([\mu^*_a+\mu_a])$, every point has
a predecessor at distance at most $\eps_a$, so $\mu^*_a+\mu_a>2n$.

For each $k\in [\mu_a]$, the thread of keys from
$[\mu^*_a+\mu_a]^{\mu_a}_k=\{x\in[\mu^*_a+\mu_a]\;|\;x=k\bmod
{\mu_a}\}$ terminates at distance $\eps^*_a$ from the successor of
$h^\downarrow_a(k)$ in $h^\downarrow_a([\mu_a])$, so the thread stays
in the same slice.  This means all keys except those in $[n]$ land in
the empty part of their slice. The same will be the case if we reach
the final key $2^\ell-1$ a key $\mu^*_a$ with
$\|h_a(\mu_a^*)\|<\eps_a$.

The keys from $[\mu^*_a+\mu_a]$ form period $0$. Generally, a period
$i>0$, starts from a key $z_i$ hashing to $(0,\eps_a)$, e.g., period
$1$ starts at $z_1=\mu^*_a+\mu_a$, and it continues until we reach key $2^\ell-1$, or till just before
we get to new key $z_{i+1}$ with $h^\downarrow_a(z_{i+1})\in
(0,\eps_a)$. This implies that $[z_i,z_{i+1})$ like $[\mu^*_a+\mu_a]$
divides intro threads, each staying within a slice between
neighboring points from $h_a^\downarrow[\mu_a]$.

Since $h^\downarrow(z_i)\in (0,\eps_a)$, for
every integer $x$, we have $h^\downarrow(z_i+x)\in (h^\downarrow(x),
h^\downarrow(x+\mu_a))$.
This implies that only the first $n-\mu_a$ elements from
$[z_i,z_{i+1})$ land between consecutive thread elements from $[n]$.
All other elements land in the empty part of their slice. It
also follows that $z_{i+1}\geq z_i+\mu^*_a$, since 
$(h^\downarrow(\mu^*_a),h^\downarrow(\mu^*_a+\mu_a))$ is the first
interval containing $0$. Hence $z_{i+1}-z_i\geq 2n-\mu_a$.

Thus, in the sequence of keys $n,...,2^\ell-1$, we first
have at least $n$ keys landing in empty parts. Next comes periods,
first with $n-\mu_a$ keys landing in filled parts, and
then at least $2n-\mu_a$ keys landing in empty parts. Eventually
we get to a last period $i$, that finishes in key $2^\ell-1$ before
reaching a key $z_{i+1}\in (0,\eps_a)$.
No matter
which key $u<2^\ell-1$, we stop at, we have that at least half the keys
in $[n,u)$ land in empty parts of slices.
\end{proof}
By Lemma \ref{lem:q-empty} we know that when $q$ is
picked randomly from $[n,u)$, then $h^\downarrow_a(q)$ lands in
the empty part of some slice with probability at least $1/2$. By
Lemma \ref{lem:b-empty}, we get $1-b^\downarrow$ in the 
empty part of the same slice with probability at least $1/(4\mu_a)$,
and this is exactly our good event. 
For fixed $a$ but random $b$ and $q$, it happened
with probability $1/(8\mu_a)$. Based on this, we will prove
\begin{lemma}\label{lem:cond-mu} For any given $\gamma\leq n/4$, uniform odd $a\in[2^\ell]$,
uniform $b\in[2^\ell]$, and uniform $q\in[u]\setminus[n]$, 
\[\Pr[h_{a,b}(q)<\min h_{a,b}([n])\mid \mu_a=\gamma]=1/(16\gamma)\]
\end{lemma}
\begin{proof}
We first note that each parameter pair $(a,b)$ has a symmetric twin
$(2^\ell-a,2^\ell-b)$ such that for every key $x$,
$h_{2^\ell-a,2^\ell-b}(x)=2^\ell-h_{a,b}(x)$. Note that $a$ odd
implies that $2^\ell-a$ is also odd, as required. The symmetry 
implies that $\mu_{2^\ell-a}=\mu_a$ while
$\eps_{2^\ell-a}=1-\eps_a$.  In particular this implies that if
we pick a uniformly odd $a$ with $\mu_a=\gamma$, then $\eps_a$ is
positive with probability exactly $1/2$.

Let us assume as we did earlier that $\eps_a$ is positive. 
Let us further assume our good event that $1-b^\downarrow$ and $h^\downarrow_a(q)$
land strictly inside the empty part of the same slice, hence that we get
no hashes from $h_{a,b}([n])$ between $0$ and $h_{a,b}(q)$. If $0$ is before
$h_{a,b}(q)$, we get $h_{a,b}(q)<\min h_{a,b}([n])$, but otherwise,
by symmetry, we get $h_{2^\ell-a,2^\ell-b}(q)<\min h_{2^\ell-a,2^\ell-b}([n])$.
Thus we have a 1-1 correspondence between the parameter choices of two
events:
\begin{itemize}
\item  parameters $a,b,q$ such that $\mu_a=\gamma$, $\eps_a$ is
positive, and $1-b^\downarrow$ and $h^\downarrow_a(q)$
land strictly inside the empty part of the same slice.
\item  parameters $a',b',q$ such that $\mu_{a'}=\gamma$, and
$h_{a',b'}(q)<\min h_{a',b'}([n])$.
\end{itemize}
In the correspondence, depending on $q$, we will either have $(a',b')=(a,b)$ or
$(a',b')=(2^\ell-a,2^\ell-b)$. The two events above are thus equally likely.

Conditioned on $\mu_a=\gamma$, we already saw that $\eps_a$ was
positive with probability $1/2$, and conditioned on that, we got
our good event with probability $1/(8\mu_a)$, for an overall
probability of $1/(16\mu_a)$. Conditioned on $\mu_{a'}=\gamma$, this
is then also the probability that $h_{a',b'}(q)<\min h_{a',b'}([n])$.
\end{proof}
We are now ready to reuse the calculations from Section \ref{sec:ms-lin} that
also defined $\mu_a$ as the smallest positive number such that 
$\|h^\downarrow_{a}(\mu_a)\|\leq 1/(2t)$.
From \req{eq:mu}, for any given odd $\gamma$ and 
uniform odd $a\in [2^\ell]$, 
\begin{eqnarray*}
\Pr[\mu_a=\gamma]&\geq&1/t-\sum_{p\ {\rm prime\ factor\ of}\ x} 2/(pt).
\end{eqnarray*}
Using Lemma \ref{lem:cond-mu}, we can
now do essentially the same calculations as
in \req{eq:calc}. For uniform odd $a\in[2^\ell]$,
uniform $b\in[2^\ell]$, and uniform $q\in[u]\setminus[n]$, 
we get
\begin{eqnarray*}
\Pr[h_{a,b}(q)<\min h_{a,b}([n])]&\geq&
\sum_{\textnormal{odd }\gamma=1}^{n/4} \Pr[\mu_a=\gamma]\,\Pr[h_{a,b}(q)<\min h_{a,b}([n])\mid \mu_a=\gamma]\\
&\geq &\sum_{\textnormal{odd }x=1}^{n/4} \left(1-2\sum_{{\rm prime\ factor}\ p{\rm\ of}\ x} 1/p\right)/(16\gamma\, t)\nonumber\\
&>&
\sum_{\textnormal{odd }\gamma=1}^{n/4} \left(1-2\sum_{{\rm prime}\ p=3,5,..} 1/p^2\right)/(16\gamma\, t)\nonumber\\
&>&
\sum_{\textnormal{odd }\gamma=1}^{n/4} 0.594\,/(16\gamma\,t)\nonumber\\
&>&H_{n/4}/(128n)=\Omega((\log n)/n).
\end{eqnarray*}
This completes the proof of Theorem \ref{thm:bad-multiply-shift-min}.
\end{proof}

\paragraph{Acknowledgments} I would like to thank some very 
thorough reviewers who came with numerous good suggestions 
for improving the presentation of this paper, including the fixing of 
several typos.

{

}

\end{document}